\RequirePackage{tikz}
\documentclass[sn-mathphys,iicol]{sn-jnl}

\jyear{2022}%

\usepackage{amsmath,amssymb,amsthm}
\usepackage{algorithm}
\usepackage{algorithmicx}
\usepackage{algpseudocode}
\usepackage{graphicx}

\usetikzlibrary{arrows,calc,positioning,shadows,shapes}
\usepackage{pgfplots}
\pgfplotsset{compat=1.6}
\usepackage{subfig}
\usepackage{textcomp}
\usepackage{xcolor}
\usepackage{multicol}
\newtheorem{theorem}{Theorem}

\newtheorem{definition}{Definition}
\newtheorem{corollary}{Corollary}

\begin{document}

\title[deepBF]{deepBF: Malicious URL detection using Self-adjusted Bloom Filter and Evolutionary Deep Learning}

\author*[1]{\fnm{Ripon} \sur{Patgiri}}\email{ripon@cse.nits.ac.in}

\author[1]{\fnm{Anupam} \sur{Biswas}}\email{anupam@cse.nits.ac.in}
\equalcont{These authors contributed equally to this work.}

\author[1]{\fnm{Sabuzima} \sur{Nayak}}\email{sabuzima\_rs@cse.nits.ac.in}
\equalcont{These authors contributed equally to this work.}

\affil*[1]{\orgdiv{Department of Computer Science \& Engineering}, \orgname{National Institute of Technology Silchar}, \orgaddress{\city{Silchar}, \postcode{788010}, \state{Assam}, \country{India}}}




\abstract{Malicious URL detection is an emerging research area due to continuous modernization of various systems, for instance, Edge Computing. In this article, we present a novel malicious URL detection technique, called deepBF (deep learning and Bloom Filter). deepBF is presented in two-fold. Firstly, we propose a self-adjusted Bloom Filter using 2-dimensional Bloom Filter. We experimentally decide the best non-cryptography string hash function. Then, we derive a modified non-cryptography string hash function from the selected hash function for deepBF by introducing biases in the hashing method and compared among the string hash functions. The modified string hash function is compared to other variants of diverse non-cryptography string hash functions. It is also compared with various filters, particularly, counting Bloom Filter, Kirsch \textit{et al.}, and Cuckoo Filter using various test cases. The test cases unearth weakness and strength of the filters. Secondly, we propose a malicious URL detection mechanism using deepBF. We apply the evolutionary convolutional neural network to identify the malicious URLs. The evolutionary convolutional neural network is trained and tested with malicious URL datasets. The output is tested in deepBF for accuracy. We have achieved many conclusions from our experimental evaluation and results and are able to reach various conclusive decisions which are presented in the article.}




\maketitle

\section{Introduction}\label{sec1}
Bloom Filter \cite{Bloom} is a famous hash data structure for membership filtering which uses a tiny amount of memory. It is known as an approximate membership filter. This tiny filter is applied in numerous research fields. For instance, BigTable \cite{BigTable} uses Bloom Filter to enhance the lookup performance. BigTable reduces unnecessary HDD access by deploying Bloom Filter. Similarly, it is deployed in various domains, namely, Big Data, Network Security \cite{DDoS,PassDB}, Network Traffic, IoT \cite{Singh}, and Bioinformatics \cite{Bio}. Besides, there are an abundant of network devices that depends on Bloom Filter. Thus, there is an immense necessity for a high accuracy Bloom Filter in Computer Networking as well as other domains because Bloom Filter can foster a system's performance and reduces the main memory consumption. 

There are diverse variants of Bloom Filters which are introduced to address several issues, for instance, counting Bloom Filter for caching URL purposes \cite{countingBF,Kirsch}. There are also similar variants of Bloom Filter, particularly, Cuckoo Filter \cite{Cuckoo}. Moreover, Patgiri \textit{et al.} introduces multidimensional Bloom Filter, called rDBF \cite{rDBF}. HFil is a high accuracy Bloom Filter extended from rDBF \cite{HFil}. Recently, a learned Bloom Filter (LBF) is introduced by M. Mitzenmacher \cite{Mitzenmacher}. LBF is currently trending in Bloom Filter. It is a combination of machine learning and Bloom Filter. Inspired from this LBF, we propose a novel technique to identify the malicious URL using evolutionary convolutional neural network (evoCNN) and Bloom Filter.  

In this article, we propose a novel self-adjusted Bloom Filter, called deepBF (Deep Learning and Bloom Filter). The complete proposed system is as follows- let, $\psi$ be a URL, $\mu\mathbb{BF}$ be the Bloom Filter to cache  malignant URL, $\beta\mathbb{BF}$ be the Bloom Filter to cache benign URLs and $\epsilon CNN$ be the evolutionary convolutional neural networks. First, a query item $\psi$ is queried to $\mu\mathbb{BF}$ for membership and if $\mu\mathbb{BF}$ returns true, then deepBF will block the URL $\psi$. Otherwise, query to $\beta\mathbb{BF}$ for membership. If $\beta\mathbb{BF}$ returns true, then the URL $\psi$ is allowed. Otherwise, $\psi$ is a new URL. Therefore, the new URL $\psi$ is input to $\epsilon CNN$ for classification. If $\epsilon CNN$ identify the URL $\psi$ as malignant, then deepBF will insert the URL $\psi$ into $\mu\mathbb{BF}$ and blocks the URL $\psi$. Otherwise, deepBF will insert the URL $\psi$ into $\beta\mathbb{BF}$ and allow the URL. This procedure reduces the load on $\epsilon CNN$ significantly. It also reduces loads on computational devices.

To achieve our proposed system, we present it in two-fold. Firstly, deepBF is designed by performing contest among the non-cryptography string hash functions in 2-Dimensional Bloom Filter (2D Bloom Filter) \cite{rDBF} using various use cases and select the best non-cryptography string hash functions. Experimental results provide the justification for not selecting cryptography string hash functions. As per our observation, the murmur2 hash function is a consistent performer and selected it to use in deepBF. The Murmur2 hash function is modified for higher performance and the resultant hash function is used in deepBF. The resultant hash function contains high biases and redundancies. However, our experimental results show that higher biases and redundancies do not affect the false positive probability (FPP) of Bloom Filter. After building a modified string hash function, deepBF is compared with Kirsch \textit{et al.} \cite{Kirsch}, counting Bloom Filter \cite{countingBF,dbloom} and Cuckoo Filter (CF) \cite{Cuckoo}. Kirsch \textit{et al.} is a modified conventional Bloom Filter, CBF is a counting Bloom Filter and CF is a similar variant of Bloom Filter. Thus, our proposed Bloom Filter is compared to prominent variant of filters. Our result shows, deepBF outperforms in different use cases. Secondly, deepBF is tested using malicious URL detection using evoCNN and proposed Bloom Filter. evoCNN is trained and tested with malicious URL dataset and we have used the dataset of \cite{Mamun} hosted in \cite{data}. The malignant and benign URLs are also tested in Bloom Filter. From this article, we have revealed strengths and weaknesses of the filters. Also, we present numerous concrete decision on Bloom Filters from our experimental results.

This article establishes preliminaries, terminologies and techniques in Section \ref{pr} which are to be used in further sections. It presents concise descriptions of Bloom Filter and its operations, and non-cryptography string hash functions. Then, provides a few related works in Section \ref{RW}. Our proposed work is described clearly through figures, equations and algorithms in Section \ref{PS}. Section \ref{ER} demonstrates the experimental environment, experimenting process and its results. Similarly, Section \ref{Ana} provides detailed analysis on our proposed systems. Likewise, a brief discussion is carried out in Section \ref{Dis}. Finally, this article is concluded with several decisions in Section \ref{Dis}.

\section{Preliminary}
\label{pr}

\subsection{Bloom Filter}
Bloom Filter is a probabilistic data structure for membership filtering capable of filtering the massive amount of data using a small memory footprint. Bloom Filter has two key issues, namely, false positives and false negatives. When a Bloom Filter avoids deletion operation, the false negative probability becomes zero, therefore, the accuracy of Bloom Filter depends on the false positive probability (FPP) of the filter. There are many variants of Bloom Filter which are introduced to reduce the issues of Bloom Filter \cite{Luo}. Also, diverse variants of Bloom Filters are introduced to address various challenges in diverse applications \cite{HLim,Fuzzy,lim}. The performance and false positive probability of Bloom Filter depend on number of hash functions. Therefore, an optimal number of hash functions are used in Bloom filter \cite{Kirsch}. If the number of hash function calls is large then it can degrade the insertion and lookup performance. If the number of hash function calls is small, then it can increase the false positive probability, but enhance the performance of insertion and lookup operations. To increase performance, we reduce the number of hash functions calls while maintaining a low false positive probability. 

Let, $\mathbb{B}$ be the Bloom Filter of size $m$ bits. The Bloom Filter has $1,~2,~3,~\ldots,~m$ cells where each cell can hold one bit, either $0$ or $1$. Let, $\mathrm{U}=\{\mathcal{K}_1,~\mathcal{K}_2,~\mathcal{K}_3,\ldots\}$ be the universe. An item $\mathcal{K}_j\in\mathrm{U}$ is mapped into Bloom Filter using $\lambda$ hash functions, let the hash functions be $\mathcal{H}_1(\mathcal{K}_j),~\mathcal{H}_2(\mathcal{K}_j),~\mathcal{H}_3(\mathcal{K}_j),~\ldots,\mathcal{H}_\lambda(\mathcal{K}_j)$. A $\lambda$ number of hash functions are invoked in insertion, deletion and query (lookup) operations. Let, $\mathcal{S}=\{\mathcal{K}_1^i,~\mathcal{K}_2^i,~\mathcal{K}_3^i,\ldots,\mathcal{K}_n^i\}$ be the inserted set of the Bloom Filter $\mathbb{B}$ where $\mathcal{S}\subset\mathrm{U}$ and $n$ is the total number of items inserted into the Bloom Filter. Let, $\mathcal{K}_i$ be the random query. The true positive, false positive, false negative and true negative are defined in Definition \ref{def1}, \ref{def2}, \ref{def3} and \ref{def4} respectively.

\begin{definition}\label{def1}
If $\mathcal{K}_i\in\mathcal{S}$ and  $\mathcal{K}_i\in\mathbb{B}$, then the result of Bloom Filter $\mathbb{B}$ is called true positive.
\end{definition}
\begin{definition}\label{def2}
If $\mathcal{K}_i\not\in\mathcal{S}$ and  $\mathcal{K}_i\in\mathbb{B}$, then the result of Bloom Filter $\mathbb{B}$ is called false positive.
\end{definition}
\begin{definition}\label{def3}
If $\mathcal{K}_i\in\mathcal{S}$ and  $\mathcal{K}_i\not\in\mathbb{B}$, then the result of Bloom Filter $\mathbb{B}$ is called false negative.
\end{definition}
\begin{definition}\label{def4}
If $\mathcal{K}_i\not\in\mathcal{S}$ and  $\mathcal{K}_i\not\in\mathbb{B}$, then the result of Bloom Filter $\mathbb{B}$ is called true negative.
\end{definition}

Bloom Filter $\mathbb{B}$ uses $m$ bits for $n$ items. Therefore, the probability of a bit to be $0$ is $(1-\frac{1}{m})$. The probability of a bit not set to $1$ using $\lambda$ hash function is
\begin{equation}\label{eq1}
    \left(1-\frac{1}{m}\right)^\lambda=\left(\left(1-\frac{1}{m}\right)^m\right)^{\frac{\lambda}{m}}=e^{-\lambda/m}
\end{equation}
where 
\begin{equation*}
    \lim_{m \to \infty} \left(1-\frac{1}{m}\right)^m=\frac{1}{e}
\end{equation*}
After insertion of $n$ items, he probability of a bit not set to $1$ is $e^{-\lambda n/m}$. Therefore, the probability of the bit to be $1$ is $1-e^{-\lambda n/m}$. Let, $\varepsilon$ be the desired false positive probability, then the all bits to be set to $1$ is 
\begin{equation}\label{eq2}
    \varepsilon=(1-e^{-\lambda n/m})^\lambda
\end{equation}
The value of $\lambda$ that minimizes false positive probability is given in Equation \eqref{eq3}.
\begin{equation}\label{eq3}
    \lambda=\frac{m}{n}ln2
\end{equation}
Replacing value of $\lambda$ and taking $ln$ in both sides in Equation \eqref{eq2}, we get
\begin{equation}\label{eq4}
    m=-\frac{n~ln~\varepsilon}{(ln~2)^2}
\end{equation}
Equation \eqref{eq4} gives us the total memory requirements for $n$ input items.

\begin{figure}[!ht]
    \centering
    \includegraphics[width=0.4\textwidth]{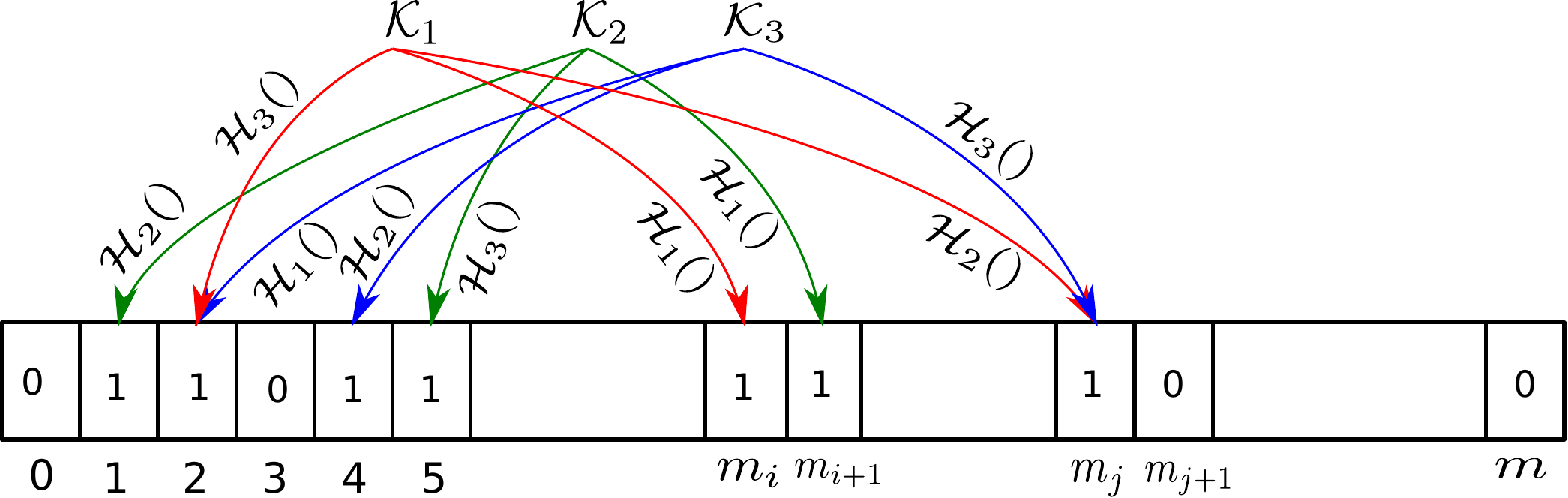}
    \caption{Mapping of $\mathcal{K}_1$, $\mathcal{K}_2$ and $\mathcal{K}_3$ into Bloom Filter using $k=3$ hash functions and these hash functions are $\mathcal{H}_1()$, $\mathcal{H}_2()$, $\mathcal{H}_3()$.}
    \label{bloom}
\end{figure}

\subsection{Operations}
Bloom Filter supports three operations, namely, insertion, deletion and query (lookup) operations. For these operations, Bloom Filter does not require complex hash functions. Instead, Bloom Filter requires the fastest non-cryptography string hash functions. Cryptography string hash function makes Bloom Filter slower, and thus, it is not wise to use MD5 and SHA2. Murmur, SuperFastHash and xxHash hash functions can be used in Bloom Filter for its operations. Bloom Filter does not require cryptography string hash function due to two reasons, namely, a) it slows down the Bloom filter performance, and b) it is unable to reduce to false positive probability. Therefore, most of the Bloom Filter uses Murmur hash functions, for instance, rDBF \cite{rDBF}. 

\subsection{Hashing Techniques}
Hashing is another factor that influences the performance of a Bloom Filter. The time complexity of the Bloom Filter operations depends on the number of hashing operations performed. 
\subsubsection{Murmur}
Murmurhash is designed by Austin Appleby in 2008 \cite{Murmur}. The name is constructed using two basic operations murmurhash perform in its inner loop, namely, multiply (MU) and rotate (R). It is a non-cryptographic hash function which helps in common hash based query. It is open to public. Various versions are also developed to improve the performance. Currently the latest version is MurmurHash3. 

\subsubsection{FNV}
Fowler/Noll/Vo (FNV) \cite{FNV} is a non-cryptography hashing technique. The technique maintains a low collision rate. FNV has high dispersion. It makes FNV suitable for hashing of similar items. In FNV, items are quickly processed while maintaining low collision rate. The cryptography hashing technique is computationally expensive to strongly prevent brute force inversion, but FNV is inexpensive. A cryptography hash function does not remain in a single state for a long time. However, in FNV hash value may be 0 and also remains in that state until a non-zero item is encountered. Moreover, when a small unpredictable item gets included in the input set FNV produces a 0 hash value, and a cryptography hash function generates a complex hash value to increase complexity, but in FNV the least significant bits of the hash value are easily visible. The available versions are FNV-1 and FNV-1a. FNV-1a performs multiply and XOR operations in a different order compared to FNV-1. This change in the order of operation resulted in better avalanche characteristics. Avalanche characteristic is a property of cryptography technique which refers to slight variation in input item heavily affects the hash value.

\subsubsection{FastHash}
FastHash \cite{fasthash} is simple non-cryptography string hash function. By default, FastHash produces 64 bits hash code. For 32 bits hash code, it deducts 32 bits code from 64 bits hash code. It is similar to Murmur hash function.

\subsubsection{CRC32}
Peterson and Brown \cite{CRC} proposed cyclic redundancy check (CRC) for error detection. It is commonly used in networking and storage devices. It helps to detect accidental alteration to data. CRC name is derived from the operations performed. The check value produced by CRC is redundancy. And, the algorithm uses cyclic codes. CRC generates a binary string of fixed length called check value. The check value is included to transmitting data. A check value is included in each data block to form a codeword. On the receiver side, again a check value is calculated for the data block or CRC is applied on whole codeword. Then, both the codewords are compared with a residue constant. In case the values differ, then data error is present in the block. CRC is used for hashing because it produces a fixed length check value. CRC32 is a 32-bit cyclic redundancy code. It returns a 32 bit long string as output. It hashes the string with less collisions. Advantages of CRC are easy implementation using a binary hardware, simple and easy mathematical analysis, and efficiently determines common errors caused by transmission channel noise.

\subsubsection{SuperfastHash}
Paul Hsieh \cite{SFH} developed a non-cryptography hash function called Superfasthash. This algorithm uses fewer instructions per input fragment. The input fragment is of 16 bits. The inner loop of the algorithm interleaves two 16 bit words. Moreover, the parameters used in the algorithm tries to give high avalanche effect.

\subsubsection{xxHash}
xxHash \cite{XXH} is a non-cryptography hashing algorithm developed by Yann Collet. It optimizes all operations to execute faster. It partition the input items into four independent streams. The responsibility of each stream is to execute block of 4 bytes per step. Each stream stores a temporary state. In the final step, all four states are combined to obtain a single state. The most important advantage of xxHash is that it’s code generator gets many opportunities to re-order opcodes to prevent delay.

\section{Related work}
\label{RW}

Kirsch \textit{et al.} proposes to reduce the number of hash functions while maintaining same FPP \cite{Kirsch}. The proposed method improves the lookup and insertion performance of Bloom Filter by reducing the number of hash functions in the conventional Bloom Filter.  Counting Bloom Filter (CBF) introduces counters for insertion and deletion operations \cite{countingBF}. Counters are decremented in deletion operations and incremented in insertion operations. It is the first variant of Bloom Filter to efficiently handle deletion operation with almost false negative free. Conventional Bloom Filter avoids deletion operation due to the false negative issue. Interestingly, CBF removes this issue using counters. However, CBF has also false negatives if counters underflow. However, this case is rare. Another kind of membership filtering is Cuckoo Filter (CF) \cite{Cuckoo}. CF uses cuckoo hashing \cite{cuckoohashing} and it is faster than Bloom Filter. 

\subsection{Learned Bloom Filter}
Learned Bloom Filter (LBF) is proposed by M. Mitzenmacher \cite{LBF} which was derived from Kraska \textit{et al.} \cite{Kraska}. LBF becomes popular from the work of M. Mitzenmacher \cite{LBF} which is generalized form. Also, M. Mitzenmacher \cite{LBF} propose sandwich structured LBF using a combination of machine learning with Bloom Filter. This structure saves time and space of a system. 

\subsection{Malicious URL}
Feng et al. \cite{Feng} use Bloom Filter to filter malicious URL. In their work, they have used multi-layer counting Bloom Filter (MCBF) for caching the malignant and benign URLs. However, deletion operation is merely used malicious URL detection. Deletion operation causes false negatives. Therefore, conventional Bloom Filter avoids deletion operation to get rid of the false negative issue.  Counting Bloom Filter (CBF) is a nearly false negative free. But, it may also occur when the counter underflows. Moreover, CBF uses higher memory footprint than conventional Bloom Filter. Dai and Shrivastava \cite{Dai} propose malicious a URL detection mechanism with M. Mitzenmacher's LBF, called Ada-BF and disjoint Ada-BF. Ada-BF is based on M. Mitzenmacher and grouping the keys to be hashed into the Bloom Filter. Based on the score, Ada-BF hashes the keys into different group in the Bloom Filter. Disjoint Ada-BF, also group keys based on score, however, the Bloom Filters are also independent, i.e., disjoint Ada-BF creates several Bloom Filters and insert the keys into a particular Bloom Filter based on the score. Both Ada-BF and disjoint Ada-BF may have skewed load. For instance, a few groups are overloaded and rest groups are under-loaded. This may happen in real life scenarios. Gerbet \textit{et al.} \cite{Gerbet} argues that non-cryptography hash functions are more vulnerable to cryptography string hash functions in Bloom Filter. We argue that this is not true for Bloom Filter. If non-cryptography hash functions are vulnerable, then cryptography hash functions are. Bloom Filter reduces hashes the keys using hash function and places the keys by modulus operations. Good string hash function may not improve the performance and FPP of Bloom Filter. Inversely, introducing more biases in the string hash function can increase the performance and reduce the FPP. On the contrary, if we use SHA or MD5, then false positive may increase and performance may also be affected adversely.  

\subsection{Evolutionary convolutional Neural Network}
Deep learning models are immensely used for numerous classification problems in different domains and proven to be superior over feature-based machine learning techniques \cite{BP}. However, the success of any deep learning model is dependent on several factors like tuning of appropriate different hyper-parameters, neural network architecture, optimizer, etc. To learn neural network weights, gradient-based optimizer such as stochastic gradient descent, min-batch gradient descent, and the Adam optimizer are widely used. However, the architecture of neural network and hyper-parameters are have to be tuned manually for better performance of the model. evoCNN models are gaining attention in recent years to overcome the manual tuning of hyper-parameters and the network architecture, (refer to detailed survey~\cite{darwish2020survey}). Currently, Several evoCNN models have been developed, mainly based on nature-inspired evolutionary optimization techniques such as Genetic Algorithm (GA), Particle Swarm Optimization (PSO), and Ant Colony Optimization (ACO). The work of Miller \textit{et al.}~\cite{miller1989designing} in 1989 was probably the first such model, which considered GA to design simple neural network. They had considered simple binary representation of neural network elements like neural units, connections, and biases etc.  Angeline \textit{et al.} ~\cite{angeline1994evolutionary} developed GA based model to construct recurrent networks. The foundation for the modern evoCNN model using GA has been laid down by Stanley and Miikkulainen~\cite{stanley2002evolving}, which learns both structure and weighting parameters of the neural network. The neural evolution follows simple feed-forward learning and mainly does three things: crossover between topologies, tracking the evolutionary units and update the topologies. Leung \textit{et al.}~\cite{leung2003tuning} proposed another model with an improved GA to further optimize the network structure considering learning of the input–output relationship. Gascón-Moreno \textit{et al.}~\cite{gascon2013evolutionary}  proposed hyperheuristic approach to adjust the number of nodes defined in each layer of the network, the number of layers, and the polynomial type. Recently,  Sun \textit{et al.}~\cite{sun2020automatically} have developed evolving deep convolutional   neural network (CNN) model using GA for evolving the architectures and connection weight initialization values to effectively address the image classification tasks.

\section{deepBF- The proposed system}
\label{PS}

We present a novel malicious URL detection mechanism, called deepBF. deepBF uses 2-dimensional Bloom Filter (2D Bloom Filter) to implement self-adjusted Bloom Filter using machine learning techniques \cite{LBF}. It deploys evolutionary deep learning technique to identify the malicious URLs. Our proposed system maintains two self-adjusted Bloom Filter, called $\mu\mathbb{BF}$ and $\beta\mathbb{BF}$ for storing malignant and benign URLs respectively. Initially, URL $\psi$ is queried to $\mu\mathbb{BF}$ and $\beta\mathbb{BF}$ to know whether $\psi$ is malignant or benign. If both filters response negative, then the URL $\psi$ is a new URL. Therefore, $\psi$ is input to evolutionary convolutional neural networks ($\epsilon CNN$) for classification. If $\epsilon CNN$ mark it as benign, then the URL $\psi$ is inserted into $\beta\mathbb{BF}$ and allow it for further processing. Otherwise, the URL $\psi$ is inserted into $\mu\mathbb{BF}$ and blocks the URL $\psi$ from further processing.

The proposed system is described in three phases; particularly, a) architecture of 2D Bloom Filter and it enhancement process, b) making 2D Bloom Filter as self-adjusted Bloom Filter, and c) the final outcome as deepBF with malicious URL detection.

\begin{figure}[!ht]
\centering
\includegraphics[width=0.45\textwidth]{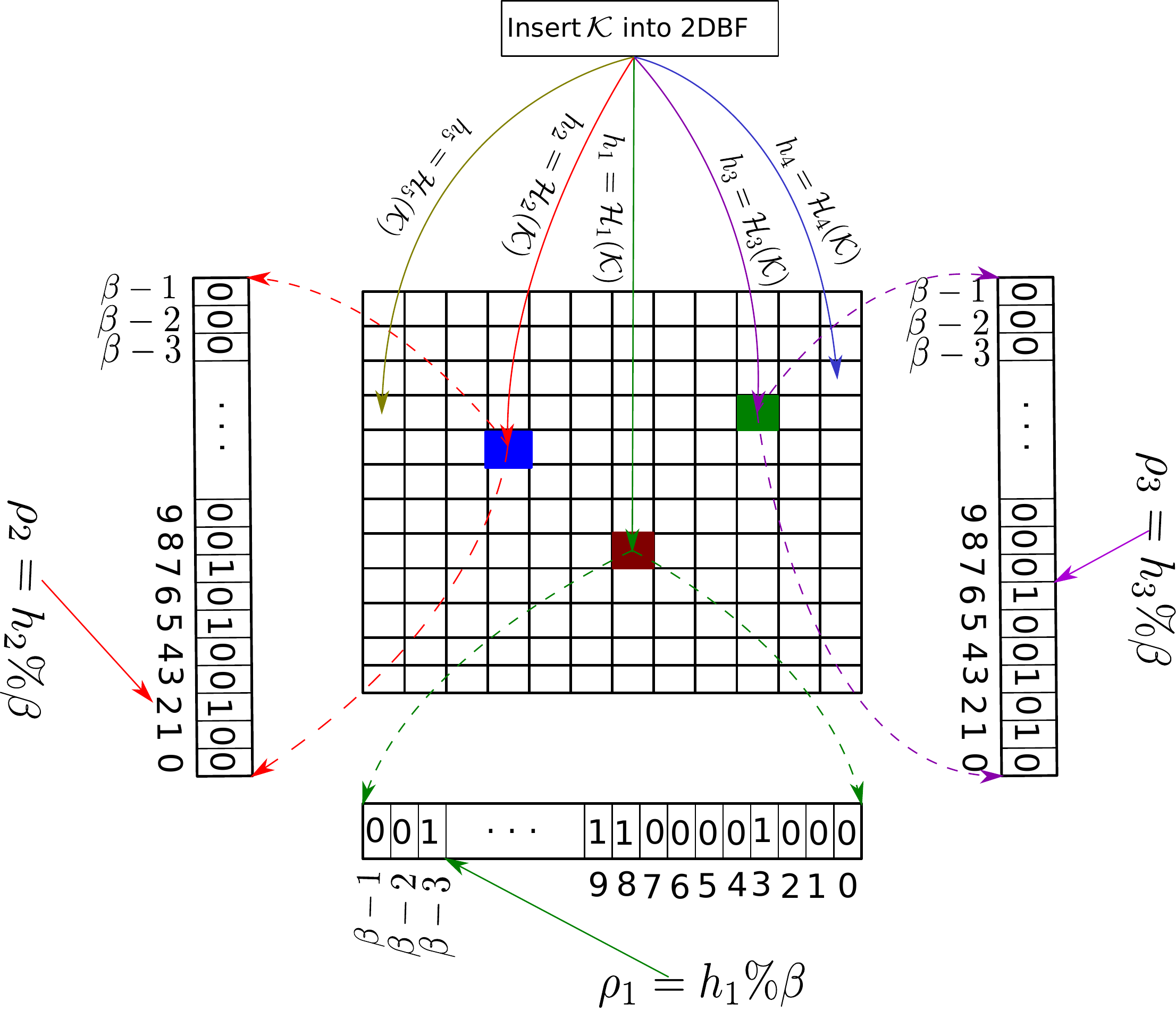}
\caption{Architecture self-adjusting Bloom Filter of deepBF depicting with five hash functions. The five hash functions are invoked for 10M items.}
\label{ins}
\end{figure}

\subsection{Insertion}
An item is inserted into self-adjusting Bloom Filter of deepBF as depicted in Figure~\ref{ins}. Algorithm \ref{algo1} implements the insertion process of self-adjusted Bloom Filter in deepBF where a set of input items is inserted into self-adjusting Bloom Filter.

\begin{algorithm}[H]
\caption{Self-adjusted Bloom Filter (2D Bloom Filter) insertion algorithm in deepBF}\label{algo1}
\footnotesize
\begin{algorithmic}[1]
\Procedure{Insertion}{$2D Bloom Filter,~File$}
    \While{$\mathcal{K}\gets ~Read~input~from~File$} 
        \State $h_1=\mathcal{H}(\mathcal{K},Seed_1)$
        \State $h_2=\mathcal{H}(\mathcal{K},Seed_2)$
        \State $h_3=\mathcal{H}(\mathcal{K},Seed_3)$
        \State $h_4=\mathcal{H}(\mathcal{K},Seed_4)$
        \State $h_5=\mathcal{H}(\mathcal{K},Seed_5)$
        \State $\Call{Insert2D Bloom Filter}{\mathcal{K},h_1}$
        \State $\Call{Insert2D Bloom Filter}{\mathcal{K},h_2}$ 
        \State $\Call{Insert2D Bloom Filter}{\mathcal{K},h_3}$ 
        \State $\Call{Insert2D Bloom Filter}{\mathcal{K},h_4}$ 
        \State $\Call{Insert2D Bloom Filter}{\mathcal{K},h_5}$
    \EndWhile
\EndProcedure
\end{algorithmic}
\end{algorithm}

deepBF uses self-adjusted Bloom Filter which is implemented using 2D Bloom Filter. Moreover, 2D Bloom Filter uses three modulus operations to place an item in a particular bit position. Let us assume, $\mathbb{B}_{M,N}$ be a 2-dimensional \textbf{unsigned long int} array to implement Bloom Filter which is initialized by zero and assuming \textbf{unsigned long int} occupies $64$ bits. The $M\not=N$ are the dimensions of the Bloom Filter and both are prime number. Equation \eqref{eq4} gives $m$, the number of memory required for $n$ items. We maintain a prime number array and the index is calculated for finding the value of $M$ and $N$. Let, $P=\{p_1,p_2,p_3,\ldots \}$ be the array of prime numbers and $i\longleftarrow \sqrt{m}$. The two dimensions are assigned by $M\longleftarrow P_{i-1}$ and $N\longleftarrow P_{i+1}$ where $i$ is a index. It is observed that the distance between two prime numbers is an important factor. It reduces the false positive rate, because the distance between $P_{i-3}$ and $P_{i+3}$ are more than the distance between $P_{i-1}$ and $P_{i+1}$. 2D Bloom Filter also requires three parameters to set a bit in $\mathbb{B}_{M,N}$, namely, $i,~j$, and $\rho$ where $\rho$ is the bit position of a particular cell, say, $\mathbb{B}_{i,j}$. The $i$ and $j$ represent particular row and column respectively. The cell size of $\mathbb{B}_{i,j}$ depends on the memory occupied by the filter for each cell, termed as $\beta$, for example, $64~bits$ for \textbf{unsigned long int}. Now, 2D Bloom Filter sets a bit in $\mathbb{B}_{i,j}$ to insert item $\mathcal{K}$  by invoking Equation \eqref{eq5}. 
\begin{equation}\label{eq5}
\mathbb{B}_{i,j}\leftarrow \mathbb{B}_{i,j}~OR~(1\ll\rho)
\end{equation}
where $OR$ is a bitwise operator and $\ll$ is the bitwise left shift operator. Now, the Murmur hash functions $\mathcal{H}(\mathcal{K})$ returns a value and assigned the returned value to $h$ by $h\leftarrow\mathcal{H}(\mathcal{K})$. To place $\mathcal{K}$, 2D Bloom Filter calculates the parameters as follows: row $i\leftarrow h\% M$, column $j \leftarrow h\%N$, and bit position $\rho\leftarrow h\% \beta$, where $\%$ is a modulus operator and $\beta$ is the bit size per cell of the Bloom Filter array. Thus, $\mathcal{K}$ is inserted using the Equation \eqref{eq5}. It is observed that $\beta=61$ have less the false positive probability than $\beta=63$ or $\beta=64$. Moreover, the number of hash functions plays critical role in reducing the false positive probability. The optimized value of number of hash functions, $\lambda$, is calculated as $\lambda=\frac{m}{n}ln2$. In our proposed systems, 2D Bloom Filter calculates the number of hash functions for achieving desired false positive probability. Therefore, 2D Bloom Filter requires $\lambda=\lceil \frac{\lambda}{2} \rceil$ hash function calls.

\subsection{Membership Query}
Similar to insertion operation, all parameters ($i,~j$ and $\rho$) are calculated for lookup operation. Equation \eqref{eq6} is invoked to query whether the item $\mathcal{K}$ is a member of 2D Bloom Filter or not.  
\begin{equation}\label{eq6}
Flag_1\leftarrow (\mathbb{B}_{i,j}~AND~(1\ll\rho))\gg \rho
\end{equation}
where $AND$ is a bitwise operator. If $Flag_1=0$, then $\mathcal{K}$ is not a member of 2D Bloom Filter.

\begin{algorithm}
\caption{2D Bloom Filter membership query of deepBF}\label{algo2} 
\footnotesize
\begin{algorithmic}[1]
\Procedure{Insertion}{$2D Bloom Filter,~File$}
    \While{$\mathcal{K}\gets ~Read~input~from~File$}
        \State $h_1=\mathcal{H}(\mathcal{K},Seed_1)$
        \State $h_2=\mathcal{H}(\mathcal{K},Seed_2)$
        \State $h_3=\mathcal{H}(\mathcal{K},Seed_3)$
        \State $h_4=\mathcal{H}(\mathcal{K},Seed_4)$
        \State $h_5=\mathcal{H}(\mathcal{K},Seed_5)$
        \If{$\Call{queryMember2D Bloom Filter}{\mathcal{K},h_1}=true$} 
            \If{$\Call{queryMember2D Bloom Filter}{\mathcal{K},h_2}=true$} 
                \If{$\Call{queryMember2D Bloom Filter}{\mathcal{K},h_3}=true$} 
                   \If{$\Call{queryMember3DBF}{\mathcal{K},h_4}=true$} 
                        \If{$\Call{queryMember2D Bloom Filterr}{\mathcal{K},h_5}=true$} 
                            \State $Found\gets Found+1$
                         \EndIf
                    \EndIf
                \EndIf
            \EndIf
        \EndIf
    \EndWhile
\EndProcedure
\end{algorithmic}
\end{algorithm}

\subsection{2D Bloom Filter as self-adjusted Bloom Filter}
Bloom Filter does not understand patterns. However, it can be trained to learn about the patterns using Machine Learning techniques. Similar to the concept of M. Mitzenmacher \cite{LBF}, we deploy evolutionary convolutional neural networks to identify the patterns and train deepBF. deepBF is deployed in malicious URL detection which is much faster than lookup in any machine learning techniques. Because, it combines both Bloom Filter and evolutionary convolutional neural networks to improve overall performance of identifying pattern. Self-adjusted Bloom Filter continuously learns about the patterns after deploying it in real project using the evolutionary convolutional neural networks. 

\begin{definition}\label{def5}
Let $\mathrm{P}$ be a pattern, and $\mathbb{B}$ is the Bloom Filter. If $\mathbb{B}$ can identify the pattern $\mathrm{P}$, then $\mathbb{B}$ is called learned Bloom Filter.
\end{definition}

Definition \ref{def5} defines the learned Bloom Filter, coined by M. Mitzenmacher \cite{LBF}. Notably, Bloom Filter does not understand the patterns. Therefore, a machine learning algorithm is required to assist the identification of patterns by the Bloom Filter. Therefore, deepBF can provide fast identification of patterns using Bloom Filter and Deep Learning method. In our proposed system, we consider Malicious URL detection as case study to validate the veracity. But deepBF can be deployed diverse applications, for instance, DDoS. As we know that Bloom Filter plays important role in the malicious URL detection. The machine learning algorithms are time consuming as compared to Bloom Filter. Moreover, the loads on a tiny device can be reduced by Bloom Filter. Also, machine learning algorithms require more memory than Bloom Filter. Therefore, Bloom Filter acts as the first layer of filtering process to reduce the load on the machine learning process. We propose a self-adjusted Bloom Filter which uses 2D Bloom Filter in deepBF. There are two situation in of 2D Bloom Filter in our proposed system; particularly, a) trained the 2D Bloom Filter before deploying it, or b) deploy 2D Bloom Filter without training it. In the both cases, deepBF works. We know that the learned Bloom Filter is trained before deploying it in a real environment. However, 2D Bloom Filter does not require any training in deepBF but it can also be trained before deploying it on real-life. Moreover, 2D Bloom Filter can be self-adjusted throughout the life-cycle which is demonstrated in Figure \ref{mal}. Therefore, our proposed 2D Bloom Filter is termed as self-adjusted Bloom Filter. Noteworthy that the $\epsilon CNN$ requires training before deploying it in real environment.

\subsection{Malicious URL Detection}
\begin{figure}[!ht]
    \centering
    \includegraphics[width=0.45\textwidth]{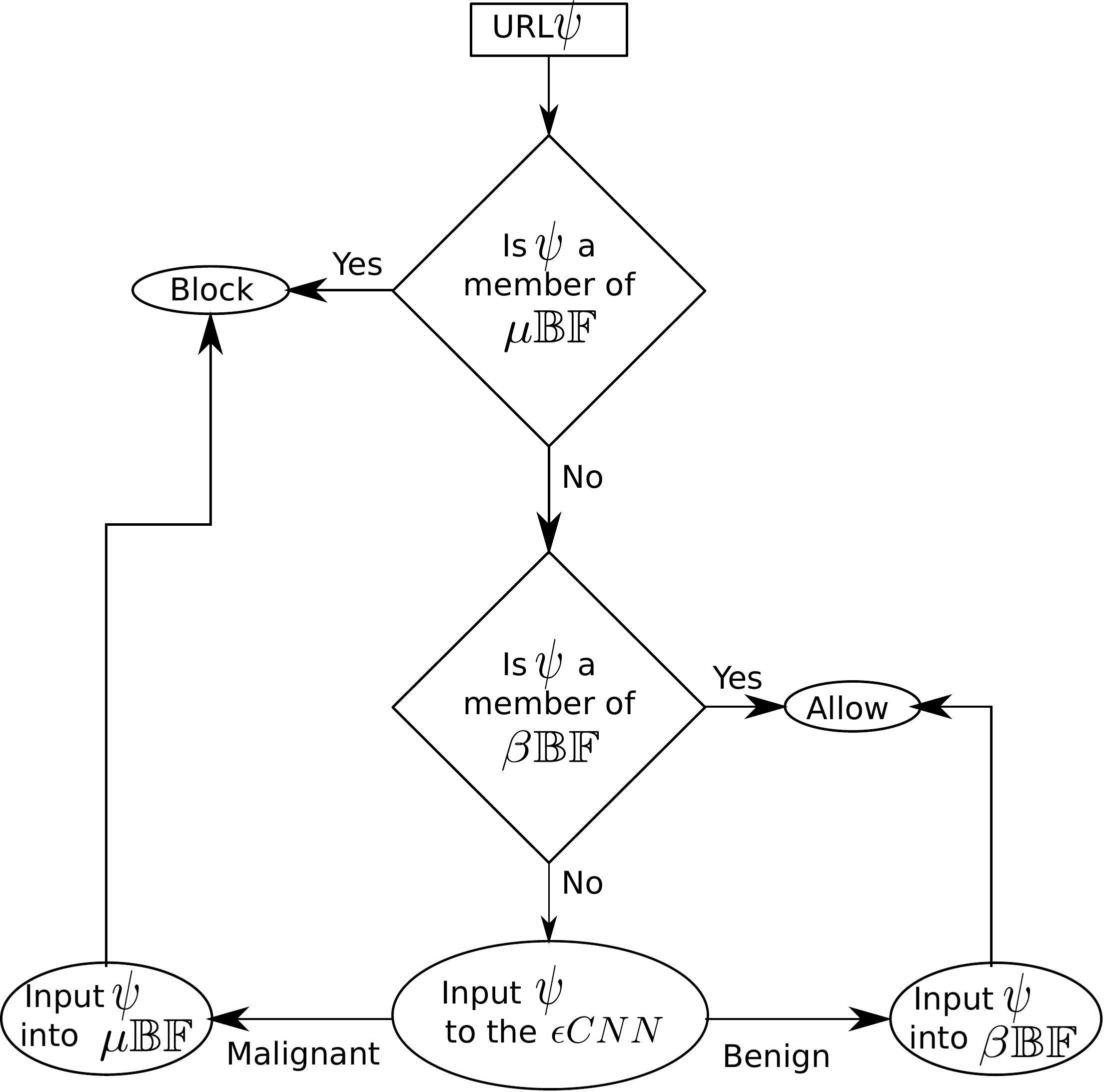}
    \caption{Malicious URL detection using two self-adjusted Bloom Filters, namely, $\mu\mathrm{BF}$ and $\beta\mathrm{BF}$ for malignant and benign URLs respectively.}
    \label{mal}
\end{figure}

Let, $\psi$ be the unknown URL, $\mu\mathbb{BF}$ and $\beta\mathbb{BF}$ be the self-adjusted Bloom Filter for malignant and benign URLs, respectively. Let $\epsilon CNN$ be the evolutionary convolutional deep learning. Figure \ref{mal} demonstrates the flow of an URL $\psi$. Firstly, $\psi$ is queried to $\mu\mathbb{BF}$ to know whether the $\psi$ is malignant or not. If $\psi$ is a member of $\mu\mathbb{BF}$, then the URL $\psi$ is blocked. Otherwise, $\psi$ is queried to $\beta\mathbb{BF}$. If $\psi$ is a member of $\beta\mathbb{BF}$, then the URL $\psi$ is benign and it is allowed; otherwise, $\psi$ is a new URL. This new URL is input into $\epsilon CNN$ for pattern recognition. The outcome of $\epsilon CNN$ is either malignant or benign. If the $\psi$ is malignant, then insert $\psi$ into  $\mu\mathbb{BF}$ and block $\psi$. Otherwise, it is inserted into $\beta\mathbb{BF}$ and the  $\psi$ is allowed. If blocked URL $\psi$ is queried for the next time, then it does not require to input into the $\epsilon CNN$ because the self-adjusted Bloom Filter blocks the URL for further processing. It saves times of the checking whether the input item is benign or malignant. Therefore, the two self-adjusted Bloom Filters grow their inputs over a time period which are much faster than the machine learning algorithms. Over a time period, only the new URLs are passed to $\epsilon CNN$ which are very less as compared to the beginning of the life-cycle of the project.


\section{Experimental Results}
\label{ER}
To evaluate our proposed system, we conduct a series of rigorous test in the low cost desktop environment. The configuration of the system is Intel(R) Core(TM) i7-7700 CPU @ 3.60GHz, Ubuntu 18.04.4 LTS with 8GiB RAM. The experimental environment is depicted in Table \ref{tab}.
\begin{table}[!ht]
    \centering
    \caption{\footnotesize{Experimental Environment Setup}}
    \begin{tabular}{p{2cm}p{4cm}}
    \hline
    \textbf{Name} & \textbf{Description}
    \\
    \hline
      CPU   &  Intel(R) Core(TM) i7-7700 CPU @ 3.60GHz\\
      L1 Cache & 32K\\
      L2 Cache & 256K\\
      RAM  &  8GB\\
      HDD & 500GB\\
      GPU & Intel® HD Graphics 630 (KBL GT2)\\
      Operating System & Ubuntu 18.04.1 LTS 64-bits\\
      \hline
    \end{tabular}
    \label{tab}
\end{table}

We present the experimental results as follows- a) selection of suitable hash function for 2D Bloom Filter, b) comparing 2D Bloom Filter with other state-of-the-art Bloom Filters, c) training and testing evolutionary convolutional Neural Network, and d) the final results of deepBF with combining 2D Bloom Filter and evolutionary convolutional Neural Network as shown in Figure \ref{mal}. 

\subsection{Test cases}
In this experimentation, we have created three different test cases to evaluate the Bloom Filter's performance. We have created three datasets, particularly, same set, mixed set and disjoint set which are defined in Definitions \ref{def7}, \ref{def8} and \ref{def9}. The size of three datasets is 10 million (10M). Initially, 10M unique keys are inserted into 2D Bloom Filter which takes $8.999744$ seconds. The same inserted keys are queried into 2D Bloom Filter which is termed as same set. The mixed set is also a unique set of items, but 50\% of query dataset items match with inserted dataset which is termed as mixed set. In disjoint set, query dataset does not match with inserted dataset. The disjoint set is a set of random keys. These test cases are used to validate the veracity of the 2D Bloom Filter in every aspect. The test cases are designed such that it can work in any kind of dataset in real environment. Most of the cases, the data are repetitive in nature; for instance, URL data. Therefore, these three test cases are enough to verify and validate the performance of a Bloom Filter in every aspect. If Bloom Filter passes these three test cases with low false positive probability, then it can withstand any kind of situation.

Interestingly, Figure \ref{hash1}  demonstrates the time measurement of 2D Bloom Filter in the three use cases. The insertion and query times are almost same for same set, however, query operation takes more times than insertion operation as shown in Figure \ref{hash2}, but the insertion operation takes more times as compared to the mixed set and disjoint set. The total false positives count is reported in Figure \ref{hash4}. 

Let, $\mathcal{S}=\{s_1,s_2,s_3,\ldots,s_m\}$ input set and input into the 2D Bloom Filter.
\begin{definition}\label{def7}
Let, $\mathcal{Q}$ is a set queried where $\mathcal{Q}=\mathcal{S}$, then the set $\mathcal{Q}$ is called same set.
\end{definition}

\begin{definition}\label{def8}
Let, $\mathcal{Q}=\{q_1,q_2\}$ be a query set where $q_1\subset\mathcal{S}$ and $q_2\cap\mathcal{S}=\phi$, then, the set $\mathcal{Q}$ is called mixed set.
\end{definition}

\begin{definition}\label{def9}
Let, $\mathcal{Q}$ be a query set where $\mathcal{Q}\cap\mathcal{S}=\phi$, then, the set $\mathcal{Q}$ is called disjoint set.
\end{definition}

\begin{definition}\label{def10}
Let, $\mathcal{Q}$ be a query set of randomly generated strings or keys, then, the set $\mathcal{Q}$ is called random set.
\end{definition}

The test cases (Definition \ref{def7}, \ref{def8}, \ref{def9} and \ref{def10}) are created to identify the strength and weakness of a Bloom Filter. The Bloom Filter does not exhibit same behavior in different test cases. Moreover, these test cases help us to evaluate the performance of the filters. We expose the strength and weakness of the filters through these test cases.

\subsection{Settings of the filters}
The required settings of the filter is $m$, $n$, $\lambda$ and $\varepsilon$. In our experiments, the desired false positive probability is $\varepsilon=0.001$ for all. From the $\varepsilon$ and $n$, the total required memory is calculated as shown in Equation \eqref{eq4}. Also, $\lambda$ can be calculated from $m$ and $n$ as shown in Equation \eqref{eq3}.

\subsection{Selection of Hash Function}
To select the best hash function for deepBF, we have conducted an extensive experiment to observe the behavior of the hash functions. We have considered eight hash functions to test the performances and accuracy, namely, FNV1, FNV1a, CRC32, Murmur2, SuperFastHash and xxHash. 2D Bloom Filter implements these hash functions to execute the insertion and query operations in 2D Bloom Filter. The best hash function is selected based on the performance of 2D Bloom Filter. The criteria for selecting the hash function to deploy in deepBF is outlined below- \newline
$\bullet$ Takes the least amount of time to process the query and insertion operation. \newline
$\bullet$ Gives high accuracy, i.e., low false positives.

\begin{definition}\label{def11}
Million operation per second (MOPS) is standard in comparison of Bloom Filter performance. It is calculated as $MOPS=\frac{n}{\tau\times 1000000}$ where $n$ is the number of items and $\tau$ is the total time taken to process the $n$ items.
\end{definition}

\pgfplotstableread[row sep=\\,col sep=&]{
interval&   Time&   	MOPS\\
MMurmur&	    1.784254&	5.60458320396087\\
Murmur2&	3.905765&	2.56031788906911\\
SuperFastHash&	3.95813&	2.52644556899344\\
xxHash& 	4.535797&  	2.2046842043416\\
CRC32&  	4.503309&	2.22058934885436\\
FastHash&	4.15934&	2.40422759380094\\
FNV1&   	4.315478&	2.31724040766747\\
FNV1a&  	4.413256&	2.2659007317953\\
}\ins
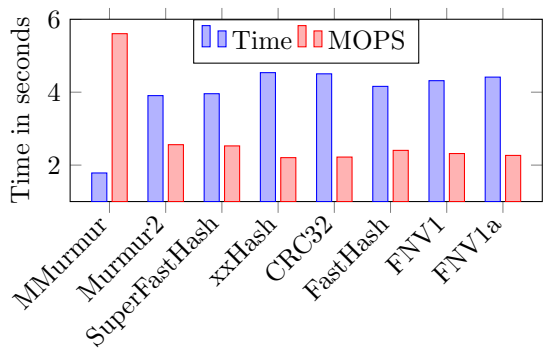
\begin{figure}[!ht]
\centering
\begin{tikzpicture}
    \begin{axis}[
            ybar,
            bar width=.2cm,
            width=0.49\textwidth,
            height=0.25\textwidth,
            legend style={at={(0.5,1)},
                anchor=north,legend columns=2,legend cell align=left},
            symbolic x coords={MMurmur,Murmur2,SuperFastHash,xxHash,CRC32,FastHash,FNV1,FNV1a},
            xtick=data,
            nodes near coords align={vertical},
            ymin=1,ymax=6,
            ylabel={Time in seconds},
            x tick label style={rotate=45,anchor=east},
        ]
        \addplot table[x=interval,y=Time]{\ins};
        \addplot table[x=interval,y=MOPS]{\ins};
        \legend{Time,MOPS}
    \end{axis}
\end{tikzpicture}
\caption{Time taken in insertion process of 10M keys into 2D Bloom Filter using various non-cryptographic string hash functions. Lower is better for Time and Higher is better for million operations per second (MOPS, Definition \ref{def10}).}
\label{hash1}
\end{figure}

The different test cases are created to evaluate the non-cryptography string hash function in 2D Bloom Filter platform. The test cases are defined in Definitions \ref{def7}, \ref{def8}, \ref{def9} and \ref{def10}. The non-cryptography hash functions are Murmur, Murmur2, SuperFastHash, xxHash, CRC32, FastHash, FNV1 and FNV1a. We have introduced more biased in Murmur2 to achieve higher speed and lower false positive probability. The modified Murmur hash function is termed as MMurmur for short. Figure \ref{hash1} depicts the insertion performance of all eight hash functions in 2D Bloom Filter platform. MMurmur with high biases is faster than rest hash functions in insertion of 10Million (10M) unique keys. MMurmur hash function is a modification and replacement of the costly operators with low-cost operators, for instance, the bitwise operators are faster than other operators. Also, number of operations are reduced. Thus, the MMurmur hash function is able to achieve higher performance than other hash functions. 

\pgfplotstableread[row sep=\\,col sep=&]{
interval&	MMurmur&	Murmur2&	SuperFastHash&	xxHash&	 CRC32&	 FastHash&	FNV1&	FNV1a\\
Same set&	1.822254&	3.600027&	3.64156&	4.235233&	4.120699&	3.838001&	4.039594&	4.136836\\
Mixed set&	1.345131&	3.148486&	3.122611&	3.802952&	3.66289&	3.201998&	3.754795&	3.702974\\
Disjoint set&	0.982423&	2.895642&	3.045309&	3.700426&	3.611164&	2.992245&	4.073009&	3.783896\\
Random set&	0.993676&	2.871343&	3.030517&	3.672954&	3.595291&	2.960768&	3.779978&	3.746942\\
}\mydata
\begin{figure}[!ht]
\centering
\begin{tikzpicture}
    \begin{axis}[
            ybar,
            bar width=.1cm,
            width=0.49\textwidth,
            height=0.2\textwidth,
            enlarge x limits=0.2,
            legend style={at={(0.5,1)},
                anchor=south,legend columns=3,legend cell align=left},
            symbolic x coords={Same set,Mixed set,Disjoint set,Random set},
            xtick=data,
            ymin=0.8,ymax=6.5,
            ylabel={Time in seconds},
             x tick label style={rotate=45,anchor=east},
        ]
        \addplot table[x=interval,y=MMurmur]{\mydata};
        \addplot table[x=interval,y=Murmur2]{\mydata};
        \addplot table[x=interval,y=SuperFastHash]{\mydata};
        \addplot table[x=interval,y=xxHash]{\mydata};
        \addplot table[x=interval,y=CRC32]{\mydata};
	    \addplot table[x=interval,y=FastHash]{\mydata};
        \addplot table[x=interval,y=FNV1]{\mydata};
	    \addplot table[x=interval,y=FNV1a]{\mydata};
        \legend{MMurmur, Murmur2, SuperFastHash, xxHash, FastHash, CRC32, FNV1, FNV1a}
    \end{axis}
\end{tikzpicture}
\caption{Time taken in lookup of 10M keys of different use cases in 2D Bloom Filter using various non-cryptographic string hash functions. Lower is better.}
\label{hash2}
\end{figure}
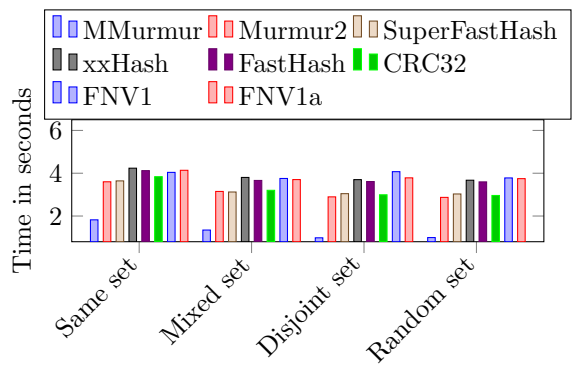

Insertion operation of Bloom Filter is not as important as lookup operation. Lookup operation is crucial in Bloom Filter because insertion operations are rare, but lookup operations are more frequent. Therefore, it is important to improve the performance of lookup operations. Figure \ref{hash2} demonstrates the performance of non-cryptography string hash function in 2D Bloom Filter platform. MMurmur hash function is at least $1.98\times$, $2.32\times$,  $2.95\times$ and $2.89\times$ faster than the other hash functions in the same set, mixed set, disjoint set and the random set respectively. Alternatively, MMurmur hash function improves at least $49.38\%$ compared to other hash functions.

\pgfplotstableread[row sep=\\,col sep=&]{
interval& 	MMurmur&	Murmur2&	SuperFastHash&	xxHash&	CRC32&	FastHash&	FNV1&	FNV1a\\
Same set&	5.48770917775458&	2.77775694460069&	2.74607585759949&	2.36114518374786&	2.42677273928525&	2.60552303139056&	2.47549629987568&	2.41730636650812\\
Mixed set&	7.43422016145639&	3.17612973346554&	3.20244820760575&	2.62953621292091&	2.73008471452869&	3.12305004562776&	2.66326124329025&	2.70053205882623\\
Disjoint set&	10.1789147851791&	3.45346558725146&	3.28373902287091&	2.70239156248497&	2.76919021124491&	3.34197233181106&	2.45518730746728&	2.64277876558975\\
Random set&	10.063642475012&	3.48269085232938&	3.29976700345189&	2.72260420359199&	2.78141602446089&	3.37750205352125&	2.6455180426976&	2.66884301918738\\
}\data
\begin{figure}[!ht]
\centering
\begin{tikzpicture}
    \begin{axis}[
            ybar,
            bar width=.1cm,
            width=0.49\textwidth,
            height=0.25\textwidth,
            enlarge x limits=0.2,
            legend style={at={(0.5,1)},
                anchor=south,legend columns=3,legend cell align=left},
            symbolic x coords={Same set,Mixed set,Disjoint set,Random set},
            xtick=data,
            nodes near coords align={vertical},
            ymin=2,ymax=14,
            ylabel={MOPS},
             x tick label style={rotate=45,anchor=east},
        ]
        \addplot table[x=interval,y=MMurmur]{\data};
        \addplot table[x=interval,y=Murmur2]{\data};
        \addplot table[x=interval,y=SuperFastHash]{\data};
        \addplot table[x=interval,y=xxHash]{\data};
        \addplot table[x=interval,y=CRC32]{\data};
	    \addplot table[x=interval,y=FastHash]{\data};
        \addplot table[x=interval,y=FNV1]{\data};
	    \addplot table[x=interval,y=FNV1a]{\data};
        \legend{MMurmur,Murmur2, SuperFastHash, xxHash, FastHash, CRC32, FNV1, FNV1a}
    \end{axis}
\end{tikzpicture}
\caption{Million Operations Per Second (MOPS) in lookup of 10M keys of different use cases in 2D Bloom Filter using various non-cryptography string hash functions. Higher is better.}
\label{hash3}
\end{figure}
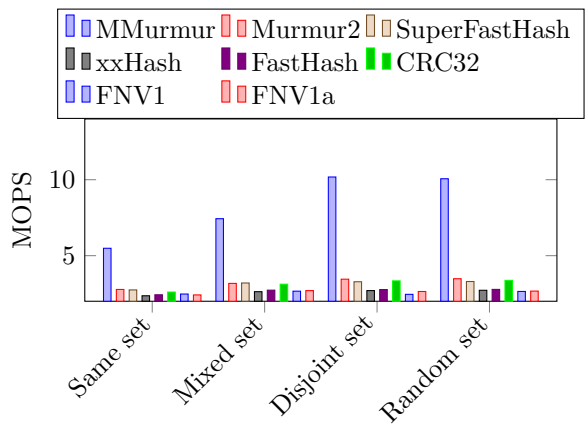

Figure \ref{hash3} illustrates performance in MOPS. MMurmur hash function outperforms all hash functions in 2D Bloom Filter platform. MMurmur hash function performs $5.48$ MOPS, $7.43$ MOPS, $10.18$ MOPS, $10.06$ MOPS in low-cost hardware for same set, mixed set, disjoint set and random set respectively. However, other hash functions perform lower MOPS than MMurmur hash function.

\pgfplotstableread[row sep=\\,col sep=&]{
interval&	MMurmur&	Murmur2&	SuperFastHash&	xxHash&	CRC32&	FastHash&	FNV1&	FNV1a\\
Mixed set&	0.00034&	0.032231&	0.025979&	0.032351&	0.029467&	0.032197&	0.031877&	0.033157\\
Disjoint set&	0.000037&	0.032284&	0.011761&	0.032234&	0.030469&	0.032193&	0.032341&	0.032163\\
Random set&	0.000036&	0.032063&	0.011655&	0.032169&	0.030415&	0.031977&	0.032032&	0.032205\\
}\fpp
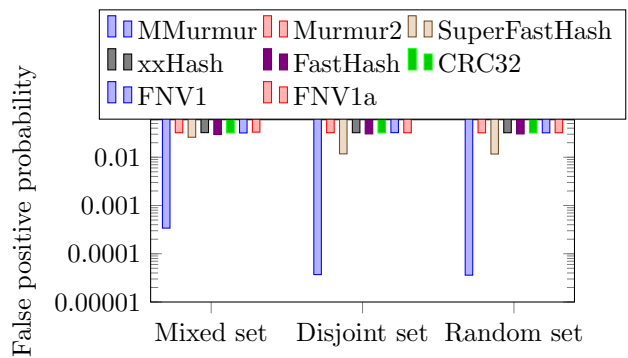
\begin{figure}[!ht]
\centering
\begin{tikzpicture}
    \begin{axis}[
            ybar,
            bar width=.1cm,
            width=0.45\textwidth,
            height=.25\textwidth,
            enlarge x limits=0.2,
            ymode=log,
            log ticks with fixed point,
            legend style={at={(0.5,1)},
                anchor=south,legend columns=3,legend cell align=left},
            symbolic x coords={Mixed set,Disjoint set,Random set},
            xtick=data,
            ymin=0.00001,ymax=0.06,
            ylabel={False positive probability},
        ]
        \addplot table[x=interval,y=MMurmur]{\fpp};
        \addplot table[x=interval,y=Murmur2]{\fpp};
        \addplot table[x=interval,y=SuperFastHash]{\fpp};
        \addplot table[x=interval,y=xxHash]{\fpp};
        \addplot table[x=interval,y=CRC32]{\fpp};
	    \addplot table[x=interval,y=FastHash]{\fpp};
        \addplot table[x=interval,y=FNV1]{\fpp};
	    \addplot table[x=interval,y=FNV1a]{\fpp};
        \legend{MMurmur,Murmur2, SuperFastHash, xxHash, FastHash, CRC32, FNV1, FNV1a}
    \end{axis}
\end{tikzpicture}
\caption{False positive probability of lookup of 10M keys of different use cases in 2D Bloom Filter using various non-cryptography string hash functions. Lower is better.}
\label{hash4}
\end{figure}

Finally, the utmost crucial factor of Bloom Filter is false positive probability and it directly proportionate to the accuracy.  Hence, Bloom Filter requires higher accuracy within desired false positive probability. The false positive probability depends on memory and the number of hash functions. Bloom Filter should not take more memory and hash functions. The number of hash function calls, reduce lookup and insertion performances. Moreover, Bloom Filter is used due to its lower memory footprint. Therefore, 2D Bloom Filter is measured in $0.001$ desired false positive probability which directly translates to $10$ hash functions calls and $17.14~MB$ primary memory consumption for $10M$ keys. However, 2D Bloom Filter allocates $17.36~MB$. Therefore, the MMurmur hash function is measured in the above mentioned settings. Notably, the false positive probability is lower than the desired false positive probability with the same settings. For all hash functions, there are no false positives for the same set. However, there are false positive probability in mixed set, disjoint set and random set. All hash functions exhibit similar false positive probability except the MMurmur hash function. MMurmur hash function exhibits extremely low false positive probability as compared to other hash functions which is depicted in Figure \ref{hash4}.

\subsection{Comparison with other filters}
With the same settings, 2D Bloom Filter is compared with other Filters, i.e., the desired false positive probability is $0.001$, the number of hash functions is $10$, the memory requirement is $17.14~MB$ or equivalent and the total $10~M$ unique keys are inserted. This article compares and demonstrates that 2D Bloom Filter with other filters that uses MMurmur hash function. 2D Bloom Filter uses five hash functions which is half of the conventional Bloom Filter.

\begin{table}[!ht]
    \centering
    \begin{tabular}{|p{2.5cm}|p{4cm}|}
    \hline
      \textbf{Bloom Filter}   & \textbf{Memory in MB} \\ \hline
       2D Bloom Filter  & 17.37 \\ \hline
       CF & 24 \\ \hline
       Kirsch \textit{et al.} & 17.14 \\ \hline
       CBF & 68.56\\ \hline
    \end{tabular}
    \caption{Memory used for $10M$ keys to achieve desired false positive probability of 0.001 by 2D Bloom Filter, CF, Kirsch \textit{et al.}, and CBF.}
    \label{tab1}
\end{table}

Table \ref{tab1} provides the total memory requirements of the filters. 2D Bloom Filter is compared with Cuckoo Filter (CF) \cite{Cuckoo,cf_src}, Kirsch \textit{et al.} \cite{Kirsch}, and counting Bloom Filter (CBF) \cite{countingBF,dbloom}. 2D Bloom Filter, CF, Kirsch, and CBF take $17.37 ~MB$, $24~MB$, $17.14~MB$ and $68.56~MB$ of memory respectively. The CBF takes higher memory than other Bloom Filters, i.e., CBF has higher false positive probability than any other Filters to achieve a desired false positive probability. If CBF or CF uses 17.14 MB memory, then both have a higher false positive probability. Alternatively, Kirsch \textit{et al.} and 2D Bloom Filter has higher accuracy.

\pgfplotstableread[row sep=\\,col sep=&]{
interval&  Time & MOPS \\
2D Bloom Filter&	1.784254& 5.60458320396087 \\
CF&         1.16536&  8.58103933548431\\
Kirsch& 	4.217597 & 2.37101837847476 \\
CBF&	    5.166789&  1.93543804478952\\
}\inss
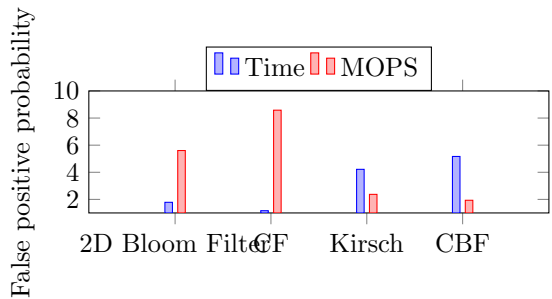
\begin{figure}[!ht]
\centering
\begin{tikzpicture}
    \begin{axis}[
            ybar,
            bar width=.1cm,
            width=0.48\textwidth,
            height=.2\textwidth,
            enlarge x limits=0.3,
            legend style={at={(0.5,1)},
                anchor=south,legend columns=2,legend cell align=left},
            symbolic x coords={2D Bloom Filter, CF, Kirsch, CBF},
            xtick=data,
            nodes near coords align={vertical},
            ymin=1,ymax=10,
            ylabel={False positive probability},
        ]
        \addplot table[x=interval,y=Time]{\inss};
        \addplot table[x=interval,y=MOPS]{\inss};
        \legend{Time,MOPS}
    \end{axis}
\end{tikzpicture}
\caption{Insertion time of 10M keys of different use cases of 2D Bloom Filter, Cuckoo Filter (CF), Kirsch \textit{et al.} and CBF. Lower is better for Time and Higher is better for MOPS.}
\label{comp1}
\end{figure}

Cuckoo filter is quite fast filter and it is faster than our proposed Bloom Filter, 2D Bloom Filter with MMurmur, and other Bloom filters in insertion. Figure \ref{comp1} demonstrates the time taken in insertions and its MOPS. CF takes less time than other Bloom Filters. Also, it's MOPS is better than other Bloom Filters.

\pgfplotstableread[row sep=\\,col sep=&]{
interval& 	2D Bloom Filter& CF&   Kirsch& 	CBF\\
Same set&	1.822254& 0.986211&	   3.334198&	4.54018\\
Mixed set&	1.345131& 0.98687&    2.65595&	3.489044\\
Disjoint set&	0.982423& 1.05343&	2.32598& 2.712857\\
Random set&	0.993676& 1.05622&   2.283228&	2.672052\\
}\look
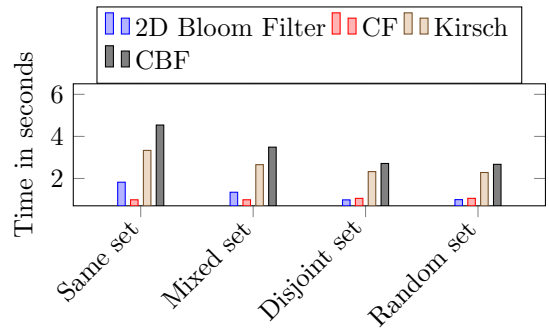
\begin{figure}[!ht]
\centering
\begin{tikzpicture}
    \begin{axis}[
            ybar,
            bar width=.1cm,
            width=0.49\textwidth,
            height=0.2\textwidth,
            enlarge x limits=0.2,
            legend style={at={(0.5,1)},
                anchor=south,legend columns=3,legend cell align=left},
            symbolic x coords={Same set,Mixed set,Disjoint set,Random set},
            xtick=data,
             x tick label style={rotate=45,anchor=east},
            nodes near coords align={vertical},
            ymin=0.7,ymax=6.5,
            ylabel={Time in seconds},
        ]
        \addplot table[x=interval,y=2D Bloom Filter]{\look};
        \addplot table[x=interval,y=CF]{\look};
        \addplot table[x=interval,y=Kirsch]{\look};
        \addplot table[x=interval,y=CBF]{\look};
        \legend{2D Bloom Filter,CF,Kirsch,CBF}
    \end{axis}
\end{tikzpicture}
\caption{Time taken in lookup of 10M keys with different use cases of 2D Bloom Filter, Cuckoo Filter (CF), Kirsch \textit{et al.} and CBF. Lower is better.}
\label{comp2}
\end{figure}

In the lookup of 10M keys, the performance of 2D Bloom Filter and CF are similar. Noteworthy that CF outperforms other Bloom Filters in same set and mixed sets. However, 2D Bloom Filter outperforms CF and other Bloom Filters in disjoint set and random set. Therefore, CF is useful in a confined environment where most of the queries are true positives and its performance is quite satisfactory, but 2D Bloom Filter is useful in random environment where most of the queries are true negatives.

\pgfplotstableread[row sep=\\,col sep=&]{
interval&	2D Bloom Filter& CF&	Kirsch& 	CBF\\
Same set&	5.48770917775458& 10.1398179497085&	2.99922200181273&	2.20255584580347\\
Mixed set&	7.43422016145639& 10.1330469058741&	3.76513112069128&	2.866114614777\\
Disjoint set&	10.1789147851791& 9.49279971141889&	4.29926310630358&	3.68615079969199\\
Random set&	10.063642475012& 9.55821911261494&	4.37976408838714&	3.7424421381021\\
}\lookm
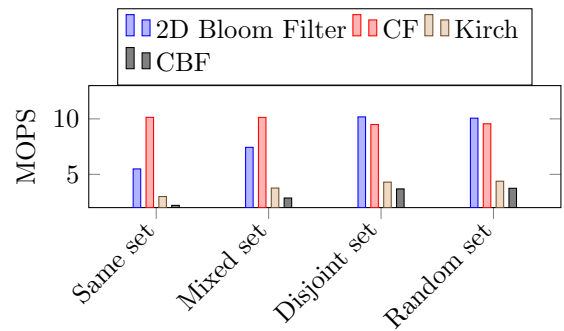
\begin{figure}[!ht]
\centering
\begin{tikzpicture}
    \begin{axis}[
            ybar,
            bar width=.1cm,
            width=0.49\textwidth,
            height=0.2\textwidth,
            enlarge x limits=0.2,
            legend style={at={(0.5,1)},
                anchor=south,legend columns=3,legend cell align=left},
            symbolic x coords={Same set,Mixed set,Disjoint set,Random set},
            xtick=data,
             x tick label style={rotate=45,anchor=east},
            nodes near coords align={vertical},
            ymin=2,ymax=13,
            ylabel={MOPS},
        ]
        \addplot table[x=interval,y=2D Bloom Filter]{\lookm};
        \addplot table[x=interval,y=CF]{\lookm};
        \addplot table[x=interval,y=Kirsch]{\lookm};
        \addplot table[x=interval,y=CBF]{\lookm};
        \legend{2D Bloom Filter,CF,Kirch,CBF}
    \end{axis}
\end{tikzpicture}
\caption{MOPS in lookup of 10M keys with different use cases in 2D Bloom Filter, CF, Kirsch \textit{et al.}, and CBF. Higher is better.}
\label{comp3}
\end{figure}

MOPS of CF is higher than other Bloom Filters in same set and mixed sets. However, 2D Bloom Filter outperforms CF and other Bloom Filters in disjoint set and random set. Undoubtedly, CF is the fastest filter, but it suffers due to kicking operation in negative queries.

\pgfplotstableread[row sep=\\,col sep=&]{
interval& 	2D Bloom Filter& CF&	Kirsch& 	CBF\\
Mixed set&	0.000340& 0.0005592&	0.001028&	0.001027\\
Disjoint set&	0.000037& 0.500877&	0.001012&	0.000996\\
Random set&	0.000036& 0.995351&	0.000989&	0.000998\\
}\fp
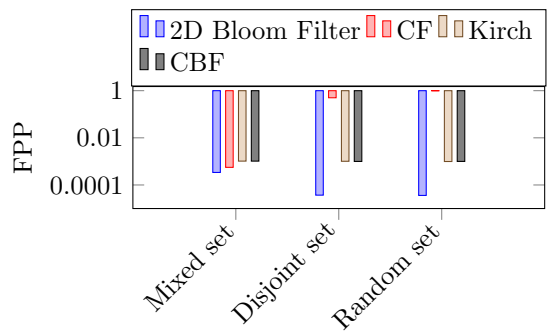
\begin{figure}[!ht]
\centering
\begin{tikzpicture}
    \begin{axis}[
            ybar,
            bar width=.1cm,
            width=0.44\textwidth,
            height=.2\textwidth,
            enlarge x limits=0.5,
            ymode=log,
            legend style={at={(0.5,1)},
                anchor=south,legend columns=3,legend cell align=left},
            symbolic x coords={Mixed set,Disjoint set,Random set},
            xtick=data,
             x tick label style={rotate=45,anchor=east},
            log ticks with fixed point,
            ymin=0.00001,ymax=1.5,
            ylabel={FPP},
        ]
        \addplot table[x=interval,y=2D Bloom Filter]{\fp};
        \addplot table[x=interval,y=CF]{\fp};
        \addplot table[x=interval,y=Kirsch]{\fp};
        \addplot table[x=interval,y=CBF]{\fp};
        \legend{2D Bloom Filter,CF,Kirch,CBF}
    \end{axis}
\end{tikzpicture}
\caption{FPP in lookup of 10M keys with different use cases in 2D Bloom Filter, CF, Kirsch \textit{et al.}, and CBF. Lower is better.}
\label{comp4}
\end{figure}

\begin{table}[!ht]
    \centering
    \caption{Accuracy of 2D Bloom Filter, CF, Kirsch \textit{et al.}, and CBF in lookup of 10M keys with different use cases. (in percentage \%)}
    \begin{tabular}{|p{1cm}|p{1.2cm}|p{1cm}|p{1cm}|p{1cm}|}
    \hline
        \textbf{Use cases}&	\textbf{2D Bloom Filter}&	\textbf{CF}& 	\textbf{Kirsch}& 	\textbf{CBF}\\ \hline
        Mixed set&	99.966& 	99.94408&	99.8972&	99.8973\\ \hline 
        Disjoint set&	99.9963&	42.51&	99.8988&	99.9004\\ \hline
        Random set&	99.9964&	0.4649&	99.9011&	99.9002\\ \hline
    \end{tabular}
    
    \label{comp5}
\end{table}

False positive rate is the most important criteria to opting a filter. All filter shows zero false positives in the same set. However, there are different false positive rate in mixed set. 2D Bloom Filter out performs all other filters in false positive rate. The false positive rate of CF in disjoint set and random set is nearly '1'. This happens due to kicking process in negative queries. Nevertheless, CF outperforms Kirsch and CBF in mixed set, but both Bloom Filter outperforms CF in disjoint set and random set as depicted in Figure \ref{comp4}. From the above benchmark, we found that CF is not suitable for some situation even though it is a fast filter. Kirch \textit{et al.} uses two Murmur2 hash function calls and the rest are manipulated better technique to reduce execution time, but still, it uses 10 hash functions for 10M items with desired false positive probability of 0.001. CBF performs moderate in all cases. However, CBF outperforms Kirsch \textit{et al.} in false positive rate. Therefore, the accuracy of 2D Bloom Filter, CF, Kirsch \textit{et al.}, and CBF are demonstrated in Table \ref{comp5}. CF exhibits lowest accuracy in disjoint set and random set.

\subsection{Evolutionary Deep Learning} 

As discussed above, the proposed malicious URL detection method consists two major components: self-adjusted Bloom Filter and evolutionary deep neural network. The self-adjusted Bloom Filter is used to block the queried URL, say $\psi$ based on its membership $\mu\mathbb{BF}$ or $\beta\mathbb{BF}$. Whereas, the evolutionary deep neural network is used to classify the newly URL $\psi$ whose membership is not defined in learn Bloom Filter. Though, deep learning models perform well in most of the classification problems, the performance depends on designing of architecture of neural network and tuning of hyper-parameters. On the other hand, evolutionary deep learning tackles both architecture and hyper-parameters of neural network. We have considered recently developed, evolutionary convolutional neural network (evoCNN)~\cite{sun2020automatically} for classifying queried new URL $\psi$. Before deployment of evoCNN, the model has to be trained on URL data.

\subsubsection{Prepossessing}
The evoCNN implemented on tensorflow platform~\cite{abadi2015tensorflow}  accepts specific shape of input dataset. Therefore, the dataset has to be processed and reshaped to fit the required input format of evoCNN. 
\begin{itemize}
    \item \textit{NaN value removal:}  Presence of NaN value in the dataset affects training of model and the model may not learn properly. Therefore, all NaN values present in the dataset is replaced with zeros.
    \item \textit{Zero padding:} Generally, the shape of input considered for the model as a square matrix. The dataset may not contain required numbers of features to rearrange those as square matrix. Therefore, additional zeros are added to complete the required shape of square matrix as shown below: 
    
    $[3,5,0,1,6,2,4]\Longrightarrow [3,5,0,1,6,2,4,0,0]$ $\longleftarrow$appended two zeros
    
    \item \textit{Input reshaping:} The evoCNN model takes $2D$ image like data to work on convolution layers. The zero padded individual instances in URL dataset is still $1D$ data, which requires to reshape into $2D$ image like data. Each instance in the URL data contains 79 features, so two zeros are appended to reshape it to $9\times 9$ matrix. In addition to this, though there has no RGB features as we have in case of colored images, still additional one dimension have to added. We considered only one channel, another dimension has to be added to this. Thus, finally each instance in URL data has been reshaped as $4D$ data. An example of $3\times 3$ to $4D$ is shown below:  
    
    \begin{equation*}
    \begin{bmatrix}
    3 & 5 & 0\\
    1 & 6 & 2\\
    4 & 0 & 0
    \end{bmatrix}
    \Longrightarrow
    \begin{bmatrix}
     \ldots 
    \begin{bmatrix}
     \ldots
    \begin{bmatrix}
    3 & 5 & 0\\
    1 & 6 & 2\\
    4 & 0 & 0
    \end{bmatrix}
    \ldots
    \end{bmatrix}
    \ldots
     \end{bmatrix}
    \end{equation*}

\end{itemize}

\subsubsection{Experimental setup}
We have considered URL dataset~\cite{Mamun,data}, which contains five different categories of URLs: spam, defacement, malware, phishing and benign. Among these first five are broadly classified as malignant. The dataset contains, separate sets of URL features for each of the four malignant categories labeled as benign or specific malignant categories. In addition, one set contains all labeled categories. All these five sets are labeled into classes malignant and benign, irrespective of their malignant category. Experimentation is done these five datasets. For training and testing of evoCNN on these five datasets different parameter values are considered as follows. Parameters related to GA are set as: number of generations 50, population size 50, and others kept default values. Parameters related to evoCNN model are set as: batch size 100, number of epochs 10, cross-entropy loss function and Adam optimizer. The maximum lengths of the convolution layers, the pooling  layers,  and  the  fully  connected  layers  are  set as same for all, i.e., $5$. For each of five datasets,  $60\%$ training, $25\%$ validation and $15\%$ testing are considered. The size of training, validation and testing for each of the datasets along with total no of samples are shown in the Table~\ref{tab:urlData}.

\begin{table}[!ht]
    \centering
    \begin{tabular}{|p{1.5cm}|p{1.2cm}|p{1.2cm}|p{1.2cm}|p{1.2cm}|}
    \hline
      \textbf{Datasets} &\textbf{\#Instances}   & \textbf{\#Training} & \textbf{\#Validation}&\textbf{\#Testing}\\ \hline
       Spam  & 14479& 8687&4923 & 869\\ \hline
       Defacement  &15711 & 9426& 5342 & 943 \\ \hline
       Malware  & 14493 & 8695 & 4928& 870\\ \hline
       Phishing  &15367 & 9220 & 5224& 923\\ \hline
       All  & 36707& 22024 & 12480 & 2203\\ \hline
 
    \end{tabular}
    \caption{Details about datasets and sizes of training, validation and testing instances.}
    \label{tab:urlData}
\end{table}

\subsubsection{URL Classification Results}
The results obtained with evoCNN for URL classification are presented in Figure~\ref{training} and Figure~\ref{testing}. The URL classification with evoCNN shows training accuracy ranging $98\%$ to $100\%$ and training loss ranging $15\%$ to $19\%$. Results on datasets with individual malignant categories as well as all combined shows high training accuracy and marginal loss. Interestingly, testing results also show high accuracy ranging $95\%$ to $98\%$ and a similar amount of loss as training. Thus, the deployment of evoCNN in the proposed architecture enables highly accurate classification of new URLs to the LBF.

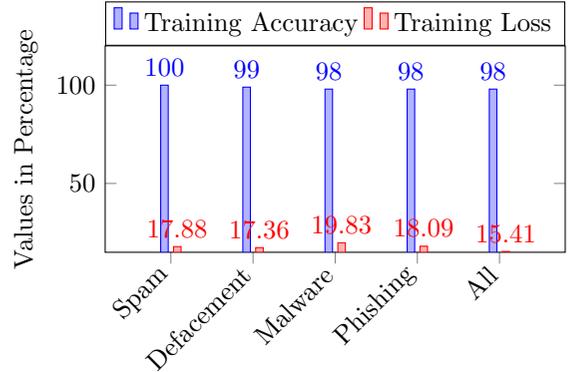
\begin{figure}[!ht]
\centering

\begin{tikzpicture}
\begin{axis}[
ybar,
bar width=.1cm,
width=0.48\textwidth,
height=0.2\textheight,
enlarge x limits = 0.2,
legend style={at={(0.5,1)},
                anchor=south,legend columns=2,legend cell align=left},
ylabel={Values in Percentage},
symbolic x coords={Spam, Defacement, Malware, Phishing, All},
xtick=data,
 x tick label style={rotate=45,anchor=east},
nodes near coords,
ymin=15,ymax=120,
]
\addplot coordinates {(Spam,100) (Defacement,99) (Malware,98) (Phishing,98) (All, 98)};
\addplot coordinates {(Spam,17.88) (Defacement,17.36) (Malware,19.83) (Phishing,18.09) (All, 15.41)};

\legend{Training Accuracy,Training Loss}
\end{axis}
\end{tikzpicture}
\caption{Training accuracy and loss of evoCNN on URL classification}
\label{training}
\end{figure}

\begin{figure}[!ht]
\centering

\begin{tikzpicture}
\begin{axis}[
ybar,
bar width=.1cm,
width=0.48\textwidth,
height=0.2\textheight,
enlarge x limits = 0.2,
legend style={at={(0.5,1)},
                anchor=south, legend columns=2,legend cell align=left},
ylabel={Values in Percentage},
symbolic x coords={Spam, Defacement, Malware, Phishing, All},
xtick=data,
 x tick label style={rotate=45,anchor=east},
nodes near coords,
ymin=15,ymax=120,
]
\addplot coordinates {(Spam,98.39) (Defacement,96.15) (Malware,97.85) (Phishing,96.29) (All, 96.35)};
\addplot coordinates {(Spam,19.37) (Defacement,17.38) (Malware,19.09) (Phishing,27.91) (All, 19.56)};

\legend{Testing Accuracy,Testing Loss}
\end{axis}
\end{tikzpicture}
\caption{Testing accuracy and loss of evoCNN on URL classification}
\label{testing}
\end{figure}
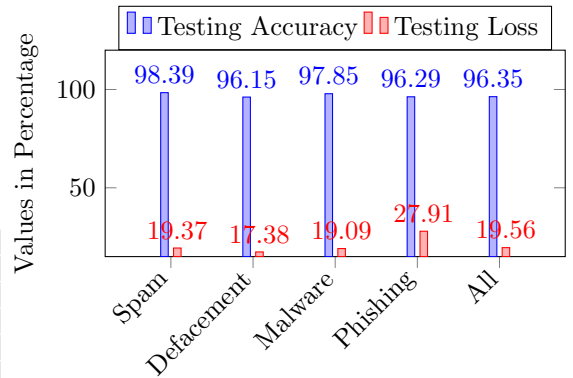

\subsection{deepBF in action}
LBF is tested using the output of the evoCNN with the dataset \cite{Mamun}. We have classified malignant and benign of all data. Therefore, there are total $129988$ malignant and $35378$ benign URLs as combined. We present this experimentation in two fold Firstly, $\mu\mathbb{BF}$ and $\beta\mathbb{BF}$ are empty. Secondly, $\mu\mathbb{BF}$ is filled with malignant URLs and tested using benign URLs.

\begin{table}[!ht]
    \centering
    \caption{Accuracy and performance testing through deduplication of malicious URLs.}
    \begin{tabular}{|p{1cm}|p{1cm}|p{1.2cm}|p{1.2cm}|p{1cm}|}
    \hline
   \textbf{Filters}&  \textbf{FPP}&	\textbf{Dedup time}&	\textbf{Accuracy}&	\textbf{Memory in KB}\\ \hline
2D Bloom Filter&	0.002523&	0.073035&	99.7477&	252.098\\ \hline
CF&	0.0000385&	0.202823&	99.996 &	488.328 \\ \hline
Kirsch &	0.071814&	0.096732&	92.8186&	228.1396 \\ \hline
CBF &	0.077876 &	0.087116 &	92.2124 &	912 \\ \hline
    \end{tabular}
    \label{mcomp}
\end{table}

Table \ref{mcomp} demonstrates performance of 2D Bloom Filter, CF, Kirsch \textit{et al.}, and CBF using deduplication of malignant URLs. In terms of accuracy, CF exhibits highest accuracy, however, it takes high memory. 2D Bloom Filter is the fastest filter in the deduplication process and CF is the slowest. Kirsch \textit{et al.} takes lowest memory while CBF consumes the highest memory.

\begin{table}[!ht]
    \centering
    \caption{Comparison of various Bloom Filter with 2D Bloom Filter for malicious URL detection by inserting malignant URLs and testing using benign URLs.}
    \begin{tabular}{|p{1cm}|p{1cm}|p{1cm}|p{1cm}|p{1cm}|p{1cm}|}
    \hline
     \textbf{Filter}& 	\textbf{FPP}&	\textbf{Insertion time}&\textbf{Lookup time}&	\textbf{Memory in KB}&	\textbf{Accuracy}\\ \hline
    \textbf{2D Bloom Filter}&	0.000283&	0.051451&	0.013258&	252.098&	99.97\\ \hline
    \textbf{CF}&	1&	0.091545&	0.02458&	488.328&	0\\ \hline
   \textbf{Kirsch}&	0.000763&	0.069181&	0.019478&	228.139&	99.92\\ \hline
    \textbf{CBF}&	0.000537&	0.044664&	0.015823&	912&	99.95\\ \hline
    \end{tabular}
    \label{mcomp1}
\end{table}

Table \ref{mcomp1} demonstrates the comparison of 2D Bloom Filter with CF, Kirsch \textit{et al.}, and CBF for false positive probability of $0.001$. In this experiment, malignant URLs are input to $\mu\mathbb{BF}$ and tested with benign URLs for accuracy. 2D Bloom Filter exhibits the lowest false positive rate and lookup time. Also, 2D Bloom Filter has highest accuracy with optimal memory sized. CBF consumed the highest memory which is $912~KB$ but exhibits the fastest insertion time. Similarly, CF also takes higher memory than 2D Bloom Filter and Kirsch \textit{et al}. CF exhibits $100\%$ false positive rate and thus its accuracy is zero. Also, it exhibits the highest insertion and lookup time. Kirsch \textit{et al.} occupies the lowest memory.

\section{Analysis}
\label{Ana}
deepBF uses 2D Bloom Filter and a cell can accommodate many input items, since, an input item occupies a single bit.  For example, \textbf{unsigned long int} occupies $8~bytes$. Therefore, the cell can retain information of at most $64$ different input items. However, it depends on the prime number $\beta$. The $\beta=64$ is not a prime number, thus, the collision probability in a cell is high. However, $\beta=61$ can lower the collision probability in a cell.

\begin{theorem}\label{th1}
Let, $\mathcal{S}=\{s_1,s_2,s_3,\ldots,s_m\}$ be the input set. Let, $\mathbb{BF}$ is the 2D Bloom Filter and $\mathcal{S}$ is inserted into $\mathbb{BF}$. 2D Bloom Filter exhibits low performance in lookup for same set.  
\end{theorem}
\begin{proof}
Same set is defined in Definition \ref{def7}. The query set $\mathcal{S}=\mathcal{Q}$. In this case, lookup process has to invoke Equation \eqref{eq6} for hash value $h_1$, $h_2$, $h_3$, $h_4$ and $h_5$ as shown in Algorithm \ref{algo2}. Invoking Equation \eqref{eq6} for all hash value are true, and hence, there are no early termination of any \textbf{IF} condition in Algorithm \ref{algo2}. Thus, it takes similar time as insertion.
\end{proof}

\begin{theorem}\label{th2}
2D Bloom Filter exhibits high performance in disjoint set.
\end{theorem}
\begin{proof}
The disjoint set is defined in Definition \ref{def9}. The necessary condition for disjoint set is $\mathcal{S}\cap\mathcal{Q}=\phi$. 2D Bloom Filter shows excellent performance in this case. Any negative query can be detected by as early as possible by \textbf{IF} condition in Algorithm \ref{algo2}. Therefore, 2D Bloom Filter terminates as early as possible if detected as negative query. Therefore, it shows excellent performance which is also shown in experimental results. 
\end{proof}

\begin{corollary}\label{col1}
2D Bloom Filter exhibits medium performance for mixed set.
\end{corollary}

Definition \ref{def8} defines a mixed set as $\mathcal{Q}=\{q_1,q_2\}$ where $q_1\subset\mathcal{S}$ and $q_2\cap\mathcal{S}=\phi$ or $q_1\cap\mathcal{S}=\phi$ and $q_2\subset\mathcal{S}$. In this case, 2D Bloom Filter exhibits medium performance which is shown in the experimental results.

\begin{theorem}\label{th3}
Let, $\zeta^\mathcal{K}$ be a cryptography string hash function of input item $\mathcal{K}$, $\varsigma^\mathcal{K}$ be the hash value of $\zeta^\mathcal{K}$, $\Upsilon^\mathcal{K}$ be the non-cryptography string hash function of input item $\mathcal{K}$ and $\upsilon^\mathcal{K}$ be the hash value of $\Upsilon^\mathcal{K}$. The performance of Bloom Filter $\mathbb{B}$ using $\upsilon^\mathcal{K}$ is higher than  $\varsigma^\mathcal{K}$.
\end{theorem}

\begin{proof}
If $\zeta^\mathcal{K}$ is MD5, SHA1 or SHA256,  then $\varsigma^\mathcal{K}$ is 128 bits, 160 bits or 256 bits long. The $\upsilon^\mathcal{K}$ can be either 32 bits or 64 bits long. In our experiment, we have used 32 bits hash functions. Therefore, $\varsigma^\mathcal{K}>\upsilon^\mathcal{K}$. The hash functions are used to distribute the keys fairly among available slots of Bloom Filter. Undoubtedly, the SHA256 or SHA512 produces strong hash values which can be used to hash the keys among the available slots. However, there is a modulus operator in hashing techniques to map a key in the slot of Bloom Filter. For instance, Bloom Filter size is $m$. Therefore, $h_\zeta=\varsigma^\mathcal{K}\% m$ should be better than $h_\Upsilon=\upsilon^\mathcal{K}\% m$. However, the ground truth differs. Firstly, $\zeta^\mathcal{K}$ is much slower than $\Upsilon^\mathcal{K}$. Secondly, $h_\zeta$ and $h_\Upsilon$ are also dependent on the value of $m$. The $m<<\varsigma^\mathcal{K}$ or $m<\upsilon^\mathcal{K}$. Therefore, the hash value is scaled under $m$ using modulus operator. The modulus operation destroys the distribution property of the hash functions. Moreover, $h_\zeta$ and $h_\Upsilon$ do not fairly distribute the keys among available Bloom Filter slots if $m$ is even number. Likewise, a MMurmur hash function has higher accuracy than Murmur hash function while the Murmur hash function is the finest non-cryptography hash function.  Therefore, the performance of Bloom Filter using $\zeta^\mathcal{K}$ lower than $\Upsilon^\mathcal{K}$. 
\end{proof}

\section{Discussion and Conclusion}
\label{Dis}

From the above experimental results, we can easily conclude that there is no requirement of the cryptography string hash function. To illustrate, the MMurmur hash function is outrun all filters where MMurmur has higher biased and redundant. Whereas, cryptography hash string hash functions have well distribution of keys. Gerbet \textit{et al.} claims that the cryptography string hash function can resist preimage and other issues. Apparently, cryptography string hash functions are not required in Bloom Filter which has been  proved experimentally in the experimental results and Theorem \ref{th3}.

Observation from the experiment, CBF has higher memory footprint issue. With the same memory footprint, conventional Bloom Filter is able to gain higher accuracy than CBF. However, CBF has a false negative free Bloom Filter provided that there is no the counter underflow. CBF is easy to handle the deletion operations of Bloom Filter. However, it occupies more memory than any other filters, that is, it has a higher false positive probability. There is a few observations in CF. First, CF is not applicable is disjoint set which is defined in Definition \ref{def9}, i.e., if the input set  and query set are disjoint, then the performance of CF degrades. Also, false positive increases. Moreover, CF consumes higher memory footprint than other variant of Bloom Filters. If CF is run again and again with the same settings, then it can crash at a point of time due to poor design of hashing. CF uses murmur2 hash function which is the finest. But the utilization of murmur2 hash function with the seed value becomes vulnerable to crash. Most importantly, the FPP is not predictable in CF. The FPP changes if CF is run again and again with the same settings. Furthermore, CF memory footprint is higher if individual key sizes are large. The memory requirements depend on the individual key size.

deepBF depends on prime numbers, for instance, the dimensions $m\not=n$ of the Bloom Filter array are prime numbers. However, deepBF is able to perform with fewer hash functions due to two modulus operations in 2D Bloom Filter, which are performed by $m$ and $n$. The key drawback of deepBF is the false positive in Bloom Filters. Particularly, if $\mu\mathbb{BF}$ returns $true$ which is a false positive. Then, the valid URL is blocked. However, the false positive probability is very less as shown in our experimental results. The deepBF comprises of two-dimensional Bloom Filter (2D Bloom Filter) and evolutionary convolutional neural network (evoCNN). deepBF uses two 2D Bloom Filter for malignant and benign URLs to filter and these two filters are first layer of the scanner. Naturally, Bloom Filters are very fast and if it is placed in the first layer of the scanner, then load on the machine is reduced. Firs, URLs are queried to the filters. If the URLs are in the 2D Bloom Filters, it saves huge times. However, if a new URL is input, then both 2D Bloom Filters returns false. Therefore, evoCNN classifies the URL as malignant or benign. Again, these URLs are inserted into the 2D Bloom Filters. Thus, 2D Bloom Filter implements learning patterns. Also, deepBF depends on evoCNN. Finally, we conclude that this work can be deployed in real world project to filter out all malignant URLs effectively and efficiently in diverse devices.

\backmatter
\section*{Statements and Declarations}
\bmhead{Competing Interests}
The research work of Dr. Anupam Biswas is supported by the Science and Engineering Board (SERB), Department of Science and Technology (DST) of the Government of India under (Grant No. EEQ/2019/000657) and (Grant No. ECR/2018/000204).






\bibliography{sn-bibliography}


\begin{thebibliography}{41}
\ifx \bisbn   \undefined \def \bisbn  #1{ISBN #1}\fi
\ifx \binits  \undefined \def \binits#1{#1}\fi
\ifx \bauthor  \undefined \def \bauthor#1{#1}\fi
\ifx \batitle  \undefined \def \batitle#1{#1}\fi
\ifx \bjtitle  \undefined \def \bjtitle#1{#1}\fi
\ifx \bvolume  \undefined \def \bvolume#1{\textbf{#1}}\fi
\ifx \byear  \undefined \def \byear#1{#1}\fi
\ifx \bissue  \undefined \def \bissue#1{#1}\fi
\ifx \bfpage  \undefined \def \bfpage#1{#1}\fi
\ifx \blpage  \undefined \def \blpage #1{#1}\fi
\ifx \burl  \undefined \def \burl#1{\textsf{#1}}\fi
\ifx \doiurl  \undefined \def \doiurl#1{\url{https://doi.org/#1}}\fi
\ifx \betal  \undefined \def \betal{\textit{et al.}}\fi
\ifx \binstitute  \undefined \def \binstitute#1{#1}\fi
\ifx \binstitutionaled  \undefined \def \binstitutionaled#1{#1}\fi
\ifx \bctitle  \undefined \def \bctitle#1{#1}\fi
\ifx \beditor  \undefined \def \beditor#1{#1}\fi
\ifx \bpublisher  \undefined \def \bpublisher#1{#1}\fi
\ifx \bbtitle  \undefined \def \bbtitle#1{#1}\fi
\ifx \bedition  \undefined \def \bedition#1{#1}\fi
\ifx \bseriesno  \undefined \def \bseriesno#1{#1}\fi
\ifx \blocation  \undefined \def \blocation#1{#1}\fi
\ifx \bsertitle  \undefined \def \bsertitle#1{#1}\fi
\ifx \bsnm \undefined \def \bsnm#1{#1}\fi
\ifx \bsuffix \undefined \def \bsuffix#1{#1}\fi
\ifx \bparticle \undefined \def \bparticle#1{#1}\fi
\ifx \barticle \undefined \def \barticle#1{#1}\fi
\bibcommenthead
\ifx \bconfdate \undefined \def \bconfdate #1{#1}\fi
\ifx \botherref \undefined \def \botherref #1{#1}\fi
\ifx \url \undefined \def \url#1{\textsf{#1}}\fi
\ifx \bchapter \undefined \def \bchapter#1{#1}\fi
\ifx \bbook \undefined \def \bbook#1{#1}\fi
\ifx \bcomment \undefined \def \bcomment#1{#1}\fi
\ifx \oauthor \undefined \def \oauthor#1{#1}\fi
\ifx \citeauthoryear \undefined \def \citeauthoryear#1{#1}\fi
\ifx \endbibitem  \undefined \def \endbibitem {}\fi
\ifx \bconflocation  \undefined \def \bconflocation#1{#1}\fi
\ifx \arxivurl  \undefined \def \arxivurl#1{\textsf{#1}}\fi
\csname PreBibitemsHook\endcsname

\bibitem{Bloom}
\begin{barticle}
\bauthor{\bsnm{Bloom}, \binits{B.H.}}:
\batitle{Space/time trade-o s in hash coding with allowable errors}.
\bjtitle{Comm. of the ACM}
\bvolume{13}(\bissue{7}),
\bfpage{422}--\blpage{426}
(\byear{1970})
\end{barticle}
\endbibitem

\bibitem{BigTable}
\begin{barticle}
\bauthor{\bsnm{Chang}, \binits{F.}},
\bauthor{\bsnm{Dean}, \binits{J.}},
\bauthor{\bsnm{Ghemawat}, \binits{S.}},
\bauthor{\bsnm{Hsieh}, \binits{W.C.}},
\bauthor{\bsnm{Wallach}, \binits{D.A.}},
\bauthor{\bsnm{Burrows}, \binits{M.}},
\bauthor{\bsnm{Chandra}, \binits{T.}},
\bauthor{\bsnm{Fikes}, \binits{A.}},
\bauthor{\bsnm{Gruber}, \binits{R.E.}}:
\batitle{Bigtable: A distributed storage system for structured data}.
\bjtitle{ACM Trans. Comput. Syst.}
\bvolume{26}(\bissue{2}),
\bfpage{4}--\blpage{1426}
(\byear{2008}).
\doiurl{10.1145/1365815.1365816}
\end{barticle}
\endbibitem

\bibitem{DDoS}
\begin{barticle}
\bauthor{\bsnm{Liu}, \binits{W.}},
\bauthor{\bsnm{Qu}, \binits{W.}},
\bauthor{\bsnm{He}, \binits{X.}},
\bauthor{\bsnm{Liu}, \binits{Z.}}:
\batitle{Detecting superpoints through a reversible counting bloom filter}.
\bjtitle{The Journal of Supercomputing}
\bvolume{63}(\bissue{1}),
\bfpage{218}--\blpage{234}
(\byear{2013}).
\doiurl{10.1007/s11227-010-0511-2}
\end{barticle}
\endbibitem

\bibitem{PassDB}
\begin{barticle}
\bauthor{\bsnm{Patgiri}, \binits{R.}},
\bauthor{\bsnm{Nayak}, \binits{S.}},
\bauthor{\bsnm{Borgohain}, \binits{S.K.}}:
\batitle{Passdb: A password database with strict privacy protocol using 3d
  bloom filter}.
\bjtitle{Information Sciences}
\bvolume{539},
\bfpage{157}--\blpage{176}
(\byear{2020}).
\doiurl{10.1016/j.ins.2020.05.135}
\end{barticle}
\endbibitem

\bibitem{Singh}
\begin{barticle}
\bauthor{\bsnm{Singh}, \binits{A.}},
\bauthor{\bsnm{Garg}, \binits{S.}},
\bauthor{\bsnm{Batra}, \binits{S.}},
\bauthor{\bsnm{Kumar}, \binits{N.}},
\bauthor{\bsnm{Rodrigues}, \binits{J.J.P.C.}}:
\batitle{Bloom filter based optimization scheme for massive data handling in
  iot environment}.
\bjtitle{Future Generation Computer Systems}
\bvolume{82}(\bissue{2018}),
\bfpage{440}--\blpage{449}
(\byear{2017}).
\doiurl{10.1016/j.future.2017.12.016}
\end{barticle}
\endbibitem

\bibitem{Bio}
\begin{barticle}
\bauthor{\bsnm{{Nayak}}, \binits{S.}},
\bauthor{\bsnm{{Patgiri}}, \binits{R.}}:
\batitle{A review on role of bloom filter on dna assembly}.
\bjtitle{IEEE Access}
\bvolume{7},
\bfpage{66939}--\blpage{66954}
(\byear{2019})
\end{barticle}
\endbibitem

\bibitem{countingBF}
\begin{barticle}
\bauthor{\bsnm{Fan}, \binits{L.}},
\bauthor{\bsnm{Cao}, \binits{P.}},
\bauthor{\bsnm{Almeida}, \binits{J.}},
\bauthor{\bsnm{Broder}, \binits{A.Z.}}:
\batitle{Summary cache: A scalable wide-area web cache sharing protocol}.
\bjtitle{IEEE/ACM Trans. Netw.}
\bvolume{8}(\bissue{3}),
\bfpage{281}--\blpage{293}
(\byear{2000}).
\doiurl{10.1109/90.851975}
\end{barticle}
\endbibitem

\bibitem{Kirsch}
\begin{barticle}
\bauthor{\bsnm{Kirsch}, \binits{A.}},
\bauthor{\bsnm{Mitzenmacher}, \binits{M.}}:
\batitle{Less hashing, same performance: Building a better bloom filter}.
\bjtitle{Random Struct. Algorithms}
\bvolume{33}(\bissue{2}),
\bfpage{187}--\blpage{218}
(\byear{2008})
\end{barticle}
\endbibitem

\bibitem{Cuckoo}
\begin{bchapter}
\bauthor{\bsnm{Fan}, \binits{B.}},
\bauthor{\bsnm{Andersen}, \binits{D.G.}},
\bauthor{\bsnm{Kaminsky}, \binits{M.}},
\bauthor{\bsnm{Mitzenmacher}, \binits{M.D.}}:
\bctitle{Cuckoo filter: Practically better than bloom}.
In: \bbtitle{Proceedings of the 10th ACM Intl. Conf. on Emerging Networking
  Experiments and Technologies}.
\bsertitle{CoNEXT '14},
pp. \bfpage{75}--\blpage{88}.
\bpublisher{IEEE},
\blocation{Sydney, Australia}
(\byear{2014}).
\doiurl{10.1145/2674005.2674994}
\end{bchapter}
\endbibitem

\bibitem{rDBF}
\begin{barticle}
\bauthor{\bsnm{Patgiri}, \binits{R.}},
\bauthor{\bsnm{Nayak}, \binits{S.}},
\bauthor{\bsnm{Borgohain}, \binits{S.K.}}:
\batitle{{rDBF}: A r-dimensional bloom filter for massive scale membership
  query}.
\bjtitle{Journal of Network and Computer Applications}
\bvolume{136},
\bfpage{100}--\blpage{113}
(\byear{2019}).
\doiurl{10.1016/j.jnca.2019.03.004}
\end{barticle}
\endbibitem

\bibitem{HFil}
\begin{bchapter}
\bauthor{\bsnm{{Patgiri}}, \binits{R.}}:
\bctitle{Hfil: A high accuracy bloom filter}.
In: \bbtitle{2019 IEEE 21st International Conference on High Performance
  Computing and Communications; IEEE 17th International Conference on Smart
  City; IEEE 5th International Conference on Data Science and Systems
  (HPCC/SmartCity/DSS)},
pp. \bfpage{2169}--\blpage{2174}
(\byear{2019})
\end{bchapter}
\endbibitem

\bibitem{Mitzenmacher}
\begin{barticle}
\bauthor{\bsnm{Mitzenmacher}, \binits{M.}}:
\batitle{Compressed bloom filters}.
\bjtitle{IEEE/ACM Trans. Netw.}
\bvolume{10}(\bissue{5}),
\bfpage{604}--\blpage{612}
(\byear{2002}).
\doiurl{10.1109/TNET.2002.803864}
\end{barticle}
\endbibitem

\bibitem{dbloom}
\begin{botherref}
\oauthor{\bsnm{Lopez}, \binits{P.}}:
Dablooms: A Scalable, Counting, Bloom Filter.
Retrieved on April, 2020 from \url{https://github.com/bitly/dablooms}
\end{botherref}
\endbibitem

\bibitem{Mamun}
\begin{bchapter}
\bauthor{\bsnm{Mamun}, \binits{M.S.I.}},
\bauthor{\bsnm{Rathore}, \binits{M.A.}},
\bauthor{\bsnm{Lashkari}, \binits{A.H.}},
\bauthor{\bsnm{Stakhanova}, \binits{N.}},
\bauthor{\bsnm{Ghorbani}, \binits{A.A.}}:
\bctitle{Detecting malicious urls using lexical analysis}.
In: \beditor{\bsnm{Chen}, \binits{J.}},
\beditor{\bsnm{Piuri}, \binits{V.}},
\beditor{\bsnm{Su}, \binits{C.}},
\beditor{\bsnm{Yung}, \binits{M.}} (eds.)
\bbtitle{Network and System Security},
pp. \bfpage{467}--\blpage{482}.
\bpublisher{Springer},
\blocation{Cham}
(\byear{2016})
\end{bchapter}
\endbibitem

\bibitem{data}
\begin{botherref}
\oauthor{\bsnm{Mamun}, \binits{M.S.I.}},
\oauthor{\bsnm{Rathore}, \binits{M.A.}},
\oauthor{\bsnm{Lashkari}, \binits{A.H.}},
\oauthor{\bsnm{Stakhanova}, \binits{N.}},
\oauthor{\bsnm{Ghorbani}, \binits{A.A.}}:
{URL} dataset {(ISCX-URL-2016)}.
Retrieved on April 2020 from
  \url{https://www.unb.ca/cic/datasets/url-2016.html}
\end{botherref}
\endbibitem

\bibitem{Luo}
\begin{barticle}
\bauthor{\bsnm{{Luo}}, \binits{L.}},
\bauthor{\bsnm{{Guo}}, \binits{D.}},
\bauthor{\bsnm{{Ma}}, \binits{R.T.B.}},
\bauthor{\bsnm{{Rottenstreich}}, \binits{O.}},
\bauthor{\bsnm{{Luo}}, \binits{X.}}:
\batitle{Optimizing bloom filter: Challenges, solutions, and comparisons}.
\bjtitle{IEEE Communications Surveys Tutorials}
\bvolume{21}(\bissue{2}),
\bfpage{1912}--\blpage{1949}
(\byear{2019})
\end{barticle}
\endbibitem

\bibitem{HLim}
\begin{barticle}
\bauthor{\bsnm{{Mun}}, \binits{J.H.}},
\bauthor{\bsnm{{Lim}}, \binits{H.}}:
\batitle{New approach for efficient ip address lookup using a bloom filter in
  trie-based algorithms}.
\bjtitle{IEEE Transactions on Computers}
\bvolume{65}(\bissue{5}),
\bfpage{1558}--\blpage{1565}
(\byear{2016})
\end{barticle}
\endbibitem

\bibitem{Fuzzy}
\begin{barticle}
\bauthor{\bsnm{{Singh}}, \binits{A.}},
\bauthor{\bsnm{{Garg}}, \binits{S.}},
\bauthor{\bsnm{{Kaur}}, \binits{K.}},
\bauthor{\bsnm{{Batra}}, \binits{S.}},
\bauthor{\bsnm{{Kumar}}, \binits{N.}},
\bauthor{\bsnm{{Choo}}, \binits{K.R.}}:
\batitle{Fuzzy-folded bloom filter-as-a-service for big data storage in the
  cloud}.
\bjtitle{IEEE Transactions on Industrial Informatics}
\bvolume{15}(\bissue{4}),
\bfpage{2338}--\blpage{2348}
(\byear{2019})
\end{barticle}
\endbibitem

\bibitem{lim}
\begin{barticle}
\bauthor{\bsnm{Lim}, \binits{H.}},
\bauthor{\bsnm{Lee}, \binits{J.}},
\bauthor{\bsnm{Byun}, \binits{H.}},
\bauthor{\bsnm{Yim}, \binits{C.}}:
\batitle{Ternary bloom filter replacing counting bloom filter}.
\bjtitle{IEEE Communications Letters}
\bvolume{21}(\bissue{2}),
\bfpage{278}--\blpage{281}
(\byear{2017}).
\doiurl{10.1109/LCOMM.2016.2624286}
\end{barticle}
\endbibitem

\bibitem{Murmur}
\begin{botherref}
\oauthor{\bsnm{Appleby}, \binits{A.}}:
MurmurHash.
Retrieved on Jan 2019 from https://sites.google.com/site/murmurhash/
(2019)
\end{botherref}
\endbibitem

\bibitem{FNV}
\begin{botherref}
\oauthor{\bsnm{Fowler}, \binits{G.}},
\oauthor{\bsnm{Noll}, \binits{L.C.}},
\oauthor{\bsnm{Vo}, \binits{K.-P.}}:
FNV Hash.
Retrieved on Aug 2019 from http://www.isthe.com/chongo/tech/comp/fnv/index.html
(2012)
\end{botherref}
\endbibitem

\bibitem{fasthash}
\begin{botherref}
\oauthor{\bsnm{Eric}}:
FastHash.
Retrieved on April 2020 from \url{https://github.com/ztanml/fast-hash}
\end{botherref}
\endbibitem

\bibitem{CRC}
\begin{barticle}
\bauthor{\bsnm{{Peterson}}, \binits{W.W.}},
\bauthor{\bsnm{{Brown}}, \binits{D.T.}}:
\batitle{Cyclic codes for error detection}.
\bjtitle{Proceedings of the IRE}
\bvolume{49}(\bissue{1}),
\bfpage{228}--\blpage{235}
(\byear{1961}).
\doiurl{10.1109/JRPROC.1961.287814}
\end{barticle}
\endbibitem

\bibitem{SFH}
\begin{botherref}
\oauthor{\bsnm{Hsieh}, \binits{P.}}:
Superfasthash.
Retrieved on Aug 2019 from http://www.azillionmonkeys.com/qed/hash.html
(2004)
\end{botherref}
\endbibitem

\bibitem{XXH}
\begin{botherref}
\oauthor{\bsnm{Collet}, \binits{Y.}}:
XXHash.
Retrieved on Aug 2019 from https://create.stephan-brumme.com/xxhash/
(2004)
\end{botherref}
\endbibitem

\bibitem{cuckoohashing}
\begin{barticle}
\bauthor{\bsnm{Pagh}, \binits{R.}},
\bauthor{\bsnm{Rodler}, \binits{F.F.}}:
\batitle{Cuckoo hashing}.
\bjtitle{Journal of Algorithms}
\bvolume{51}(\bissue{2}),
\bfpage{122}--\blpage{144}
(\byear{2004})
\end{barticle}
\endbibitem

\bibitem{LBF}
\begin{bchapter}
\bauthor{\bsnm{Mitzenmacher}, \binits{M.}}:
\bctitle{A model for learned bloom filters and optimizing by sandwiching}.
In: \beditor{\bsnm{Bengio}, \binits{S.}},
\beditor{\bsnm{Wallach}, \binits{H.}},
\beditor{\bsnm{Larochelle}, \binits{H.}},
\beditor{\bsnm{Grauman}, \binits{K.}},
\beditor{\bsnm{Cesa-Bianchi}, \binits{N.}},
\beditor{\bsnm{Garnett}, \binits{R.}} (eds.)
\bbtitle{Advances in Neural Information Processing Systems 31},
pp. \bfpage{464}--\blpage{473}.
\bpublisher{Curran Associates, Inc.}, \blocation{???}
(\byear{2018})
\end{bchapter}
\endbibitem

\bibitem{Kraska}
\begin{bchapter}
\bauthor{\bsnm{Kraska}, \binits{T.}},
\bauthor{\bsnm{Beutel}, \binits{A.}},
\bauthor{\bsnm{Chi}, \binits{E.H.}},
\bauthor{\bsnm{Dean}, \binits{J.}},
\bauthor{\bsnm{Polyzotis}, \binits{N.}}:
\bctitle{The case for learned index structures}.
In: \bbtitle{Proceedings of the 2018 International Conference on Management of
  Data}.
\bsertitle{SIGMOD ’18},
pp. \bfpage{489}--\blpage{504}.
\bpublisher{Association for Computing Machinery},
\blocation{New York, NY, USA}
(\byear{2018}).
\doiurl{10.1145/3183713.3196909}.
\burl{https://doi.org/10.1145/3183713.3196909}
\end{bchapter}
\endbibitem

\bibitem{Feng}
\begin{bchapter}
\bauthor{\bsnm{{Feng}}, \binits{Y.}},
\bauthor{\bsnm{{Huang}}, \binits{N.}},
\bauthor{\bsnm{{Chen}}, \binits{C.}}:
\bctitle{An efficient caching mechanism for network-based url filtering by
  multi-level counting bloom filters}.
In: \bbtitle{2011 IEEE International Conference on Communications (ICC)},
pp. \bfpage{1}--\blpage{6}
(\byear{2011})
\end{bchapter}
\endbibitem

\bibitem{Dai}
\begin{botherref}
\oauthor{\bsnm{Dai}, \binits{Z.}},
\oauthor{\bsnm{Shrivastava}, \binits{A.}}:
Adaptive Learned Bloom Filter (Ada-BF): Efficient Utilization of the Classifier
(2019)
\end{botherref}
\endbibitem

\bibitem{Gerbet}
\begin{bchapter}
\bauthor{\bsnm{{Gerbet}}, \binits{T.}},
\bauthor{\bsnm{{Kumar}}, \binits{A.}},
\bauthor{\bsnm{{Lauradoux}}, \binits{C.}}:
\bctitle{The power of evil choices in bloom filters}.
In: \bbtitle{2015 45th Annual IEEE/IFIP International Conference on Dependable
  Systems and Networks},
pp. \bfpage{101}--\blpage{112}
(\byear{2015})
\end{bchapter}
\endbibitem

\bibitem{BP}
\begin{barticle}
\bauthor{\bsnm{{Pourbabaee}}, \binits{B.}},
\bauthor{\bsnm{{Roshtkhari}}, \binits{M.J.}},
\bauthor{\bsnm{{Khorasani}}, \binits{K.}}:
\batitle{Deep convolutional neural networks and learning ecg features for
  screening paroxysmal atrial fibrillation patients}.
\bjtitle{IEEE Transactions on Systems, Man, and Cybernetics: Systems}
\bvolume{48}(\bissue{12}),
\bfpage{2095}--\blpage{2104}
(\byear{2018})
\end{barticle}
\endbibitem

\bibitem{darwish2020survey}
\begin{barticle}
\bauthor{\bsnm{Darwish}, \binits{A.}},
\bauthor{\bsnm{Hassanien}, \binits{A.E.}},
\bauthor{\bsnm{Das}, \binits{S.}}:
\batitle{A survey of swarm and evolutionary computing approaches for deep
  learning}.
\bjtitle{Artificial Intelligence Review}
\bvolume{53}(\bissue{3}),
\bfpage{1767}--\blpage{1812}
(\byear{2020})
\end{barticle}
\endbibitem

\bibitem{miller1989designing}
\begin{bchapter}
\bauthor{\bsnm{Miller}, \binits{G.F.}},
\bauthor{\bsnm{Todd}, \binits{P.M.}},
\bauthor{\bsnm{Hegde}, \binits{S.U.}}:
\bctitle{Designing neural networks using genetic algorithms.}
In: \bbtitle{ICGA},
vol. \bseriesno{89},
pp. \bfpage{379}--\blpage{384}
(\byear{1989})
\end{bchapter}
\endbibitem

\bibitem{angeline1994evolutionary}
\begin{barticle}
\bauthor{\bsnm{Angeline}, \binits{P.J.}},
\bauthor{\bsnm{Saunders}, \binits{G.M.}},
\bauthor{\bsnm{Pollack}, \binits{J.B.}}:
\batitle{An evolutionary algorithm that constructs recurrent neural networks}.
\bjtitle{IEEE transactions on Neural Networks}
\bvolume{5}(\bissue{1}),
\bfpage{54}--\blpage{65}
(\byear{1994})
\end{barticle}
\endbibitem

\bibitem{stanley2002evolving}
\begin{barticle}
\bauthor{\bsnm{Stanley}, \binits{K.O.}},
\bauthor{\bsnm{Miikkulainen}, \binits{R.}}:
\batitle{Evolving neural networks through augmenting topologies}.
\bjtitle{Evolutionary computation}
\bvolume{10}(\bissue{2}),
\bfpage{99}--\blpage{127}
(\byear{2002})
\end{barticle}
\endbibitem

\bibitem{leung2003tuning}
\begin{barticle}
\bauthor{\bsnm{Leung}, \binits{F.H.-F.}},
\bauthor{\bsnm{Lam}, \binits{H.-K.}},
\bauthor{\bsnm{Ling}, \binits{S.-H.}},
\bauthor{\bsnm{Tam}, \binits{P.K.-S.}}:
\batitle{Tuning of the structure and parameters of a neural network using an
  improved genetic algorithm}.
\bjtitle{IEEE Transactions on Neural networks}
\bvolume{14}(\bissue{1}),
\bfpage{79}--\blpage{88}
(\byear{2003})
\end{barticle}
\endbibitem

\bibitem{gascon2013evolutionary}
\begin{barticle}
\bauthor{\bsnm{Gasc{\'o}n-Moreno}, \binits{J.}},
\bauthor{\bsnm{Salcedo-Sanz}, \binits{S.}},
\bauthor{\bsnm{Saavedra-Moreno}, \binits{B.}},
\bauthor{\bsnm{Carro-Calvo}, \binits{L.}},
\bauthor{\bsnm{Portilla-Figueras}, \binits{A.}}:
\batitle{An evolutionary-based hyper-heuristic approach for optimal
  construction of group method of data handling networks}.
\bjtitle{Information Sciences}
\bvolume{247},
\bfpage{94}--\blpage{108}
(\byear{2013})
\end{barticle}
\endbibitem

\bibitem{sun2020automatically}
\begin{botherref}
\oauthor{\bsnm{Sun}, \binits{Y.}},
\oauthor{\bsnm{Xue}, \binits{B.}},
\oauthor{\bsnm{Zhang}, \binits{M.}},
\oauthor{\bsnm{Yen}, \binits{G.G.}},
\oauthor{\bsnm{Lv}, \binits{J.}}:
Automatically designing cnn architectures using the genetic algorithm for image
  classification.
IEEE Transactions on Cybernetics
(2020)
\end{botherref}
\endbibitem

\bibitem{cf_src}
\begin{botherref}
\oauthor{\bsnm{Fan}, \binits{B.}}:
cuckoofilter.
Retrieved on April 2020 from \url{https://github.com/efficient/cuckoofilter}
\end{botherref}
\endbibitem

\bibitem{abadi2015tensorflow}
\begin{botherref}
\oauthor{\bsnm{Abadi}, \binits{M.}},
\oauthor{\bsnm{Agarwal}, \binits{A.}},
\oauthor{\bsnm{Barham}, \binits{P.}},
\oauthor{\bsnm{Brevdo}, \binits{E.}},
\oauthor{\bsnm{Chen}, \binits{Z.}},
\oauthor{\bsnm{Citro}, \binits{C.}},
\oauthor{\bsnm{Corrado}, \binits{G.S.}},
\oauthor{\bsnm{Davis}, \binits{A.}},
\oauthor{\bsnm{Dean}, \binits{J.}},
\oauthor{\bsnm{Devin}, \binits{M.}}, et al.:
Tensorflow: Large-scale machine learning on heterogeneous systems
(2015)
\end{botherref}
\endbibitem

\end{thebibliography}

\end{document}